\documentclass[twocolumn,5p]{elsarticle}

\usepackage{eulervm}

\usepackage{graphicx}
\usepackage{array}
\usepackage{amsmath}
\usepackage{amssymb}
\usepackage{amsthm}
\usepackage{mathrsfs}
\usepackage{bbm}
\usepackage{bbold}
\usepackage{xfrac}
\usepackage{algorithmic}
\usepackage{algorithm}
\usepackage{upgreek}
\usepackage{multirow}
\usepackage{arydshln}
\usepackage{subfigure}

\usepackage{amssymb}

\journal{Elsevier Ad Hoc Networks Journal}
\DeclareMathOperator{\argmax}{\arg\max}

\begin{document}

\begin{frontmatter}



\title{MPAR: A Movement Pattern-Aware Optimal Routing for Social Delay Tolerant Networks}


\author{Lei~You$^{1}$,
        Jianbo~Li$^{2*}$, 
        Changjiang Wei$^{3}$ 
        Lejuan Hu$^{4}$
}

\address{Information Engineering College, Qingdao University, Qingdao 266071, Shandong Province, China \\
\texttt{youleiqdu@gmail.com$^1$,lijianboqdu@gmail.com$^2$,wcj@qdu.edu.cn$^3$, 907316924@qq.com$^4$}
}

\begin{abstract}
Social Delay Tolerant Networks (SDTNs) are a special kind of Delay Tolerant Network (DTN) that consists of a number of mobile devices with social characteristics. The current research achievements on routing algorithms tend to separately evaluate the available profit for each prospective relay node and cannot achieve the global optimal performance in an overall perspective. In this paper, we propose a Movement Pattern-Aware optimal Routing (MPAR) for SDTNs, by choosing the optimal relay node(s) set for each message, which eventually based on running a search algorithm on a hyper-cube solution space. Concretely, the movement pattern of a group of node(s) can be extracted from the movement records of nodes. Then the set of commonly visited locations for the relay node(s) set and the destination node is obtained, by which we can further evaluate the co-delivery probability of the relay node(s) set. Both local search scheme and tabu-search scheme are utilized in finding the optimal set, and the tabu-search based routing Tabu-MPAR is proved able to guide the relay node(s) set in evolving to the optimal one. We demonstrate how the MPAR algorithm significantly outperforms the previous ones through extensive simulations, based on the synthetic SDTN mobility model.
\end{abstract}

\begin{keyword}


Delay Tolerant Networks, Mobile Social Networks, Movement,  Opportunistic routing
\end{keyword}

\end{frontmatter}

\theoremstyle{plain}
\newtheorem{definition}{Definition}
\newtheorem{theorem}{Theorem}
\newtheorem{lemma}{Lemma}
\newtheorem{proposition}{Proposition}
\newtheorem{postulation}{Postulation}
\section{Introduction}
In traditional data networks such as Internet, there are usually some assumptions of the network model e.g. the existence of at least one end-to-end path between source-destination pair \cite{Cao2012}. Any arbitrary link connecting two nodes is assumed to be bidirectional supporting symmetric data rates with low error probability and latency \cite{Vastardis2013}. Messages are buffered in intermediate nodes (e.g. routers) and further forwarded to the next-hop relay or successfully received by the destination. In this case, each message is not expected to occupy the buffer of nodes for a long period of time. However, these all above usually fail in the context of Delay Tolerant Networks (DTNs)\cite{Fall2003}. Some applications, e.g., email service, address the delivery success while having relatively flexible requirement of latency, which is known as ``delay-tolerant". For further popularizing these kind of applications, we have to reconsider the widely used network architecture so as to relax the assumption of the continuous end-to-end connectivity that TCP/IP based \cite{Ott2012}. 

Recently years there have seen the wide adoptions of mobile devices such as smartphones, laptops and tablet PCs. The mobile devices may form a network in an ad hoc manner, to work as an auxiliary to cellular networks for some services such as social information sharing. Many researchers use the term ``Social Delay Tolerant Networks (SDTNs)'' to describe this special kind of DTN, where mobile users move around and communicate with each other via their carried short-distance wireless communication devices \cite{Wei2013b,Bigwood2009,Nguyen2009}. Since SDTNs experience intermittent connectivity incurred by the mobility of users, routing is still the most challenging problem \cite{Wei2013} for such networks. The fundamental characteristics of STDN is that there exist potential social relationships behind the nodes in the network, which in turn affects the movement pattern of nodes. Many previous works have explored how to leverage the social relationships so as to enhance the routing performance. For example, social-aware routing algorithms based on social network analysis have been proposed, such as Bubble Rap \cite{Hui2008}, SimBet \cite{Daly2007}. Besides, algorithms that based on some custom defined social metrics are proposed in \cite{Xiao2013,Li2013,Ning2013,Wu2012,Moreira2012,Ciobanu2013}.

In many real SDTNs, mobile users that have a common interest generally will visit some locations that are related to this interest. Previous works observed that 50\% of mobile users in this network spent 74.0\% of their time at a single access point (AP) to show the characteristic of frequently visiting a few locations \cite{Henderson2004a}. This indicates the feasibility to ``connect'' some nodes through their commonly visited locations where some cache devices are deployed for buffering the message, e.g. throwbox proposed in \cite{Ibrahim2009}. Cache devices have several advantages over ordinary nodes. First, there would usually be some stable connections between nodes and the cache device \cite{Henderson2004a}, which hardly exist between two ordinary nodes due to their mobility. Second, they are fixed as a kind of network infrastructures, and their capacity can be much more sufficient than the ordinary nodes. 
From the above discussion, intuitively we can deploy cache devices in frequently visited locations to improve the routing performance in SDTNs. 

In this paper, we focus on the multi-copy unicast routing problem in SDTNs, where there is one source-destination nodes pair for each packet, and there can be multiple replicas of the packet in the network. 
Different from previous works that based on classical social network analysis or custom defined social metrics, our scheme makes use of the movement records collected by nodes. The basic idea is to view the set of nodes holding a packet or its replicas as an entirety, thus extracting its frequently visited locations from the movement records. The frequently visited locations of the destination node is obtained in the same way. Then we get the intersection between the two locations sets and then derive the co-delivery probability for the whole relay node set. Our goal is to calculate the optimal nodes set that maximizes the co-delivery probability. In addition, since that each generated message is only valid before reaching its deadline, we take the message time-to-live value into consideration when extracting the movement pattern from the movement records. Under this model, we propose the Movement Pattern-Aware optimal Routing (MPAR) for SDTNs. We adopt the optimal opportunistic routing scheme by maintaining an optimal node set for the destination node. To the best of our knowledge, this is the first work to exploit the social-aware routing algorithm that leverages the movement pattern of a group of nodes. Our main contributions are summarized as follows:
\begin{enumerate}
\item We present a periodical time-aware movement record model and extract the movement pattern from the movement record of nodes. Each node set is viewed as an entirety during the whole routing process. Corresponding Movement pattern can be extracted from the same set in different time intervals. The frequently visited locations are directly corresponding to the movement pattern for each node(s) set.

\item Under our network model that cache devices are deployed in several positions, we analyze two key performance metrics, based on which the routing problem is formally proved to be an $NP-Hard$ combinatorial optimization problem.

\item Two search algorithms are proposed to solve the optimization problem, which are respectively based on the local search scheme and the tabu search scheme. The key elements of the tabu search scheme are specifically defined for our optimal set search problem.

\item Two respective movement pattern-aware routing schemes are designed based on the local search algorithm and the tabu search algorithm, which is called Local-MPAR and Tabu-MPAR and are respectively corresponding to the reactive routing strategy and the proactive routing strategy. Tabu-MPAR is proved able to guide the relay node(s) set in evolving to the optimal one.
\end{enumerate}

The rest of this paper is organized as follows. In section \ref{sec:preliminaries} we introduce the system model. In section \ref{sec:overview} we analyze some key properties behind the routing, and model the routing as an optimal search problem. We discuss the computational hardness of this problem and propose a heuristic method to approximately obtain the optimal solution in section \ref{sec:analysis}. The details about the routing algorithm are given in section \ref{sec:routing}. Section~\ref{sec:simulation} analyses the simulation result. We conclude the paper in Section \ref{sec:conclusion}. The proofs of Theorem 1 and 2 are presented in the Appendix.

\section{Preliminaries}
\label{sec:preliminaries}
In this section, we give the system model and the assumptions behind our scheme. 

\subsection{Network Model}

We consider a SDTN composed of $n$ nodes $\overline{N}=\{n_{i}|n_{i}\in \overline{N}, 1\leq i\leq n\}$ moving among $m$ locations $\overline{A}=\{a_{j}|a_{j}\in\overline{N},1\leq j\leq m\}$. Each mobile node frequently visits some locations $A\subseteq \overline{A}$. This model is derived from the case of real mobile networks. A typical example is the Wi-Fi campus network at Dartmouth College. Another SDTN that follows this characteristic is the Vehicular Ad-hoc NETworks (VANETs), where lots of buses and taxies move among bus stations and taxi stops. 
Like in the previous work \cite{Xiao2013,Gao2009}, we assume that the behavior of each node $n_i$ visiting any location $a_j\in A(n_i)$ follows the Poisson process.
In other words, the time interval that each node visits a location follows an exponential distribution.
Besides, each location in our model is assumed to have a throwbox\cite{Ibrahim2009} to store and transmit messages. And we assume that the capacity of the throwbox is sufficient enough to take custody of messages. 

We assume that each message has a time-to-live value $\tau _l$, indicating the remaining life time of the message. With the time passing, $\tau _l$ gradually decreases, and when $\tau _l$ is zero, the message is deleted from nodes' buffer. This time-to-live field prevents a message from staying in the network for a very long time, and allows the application to assign a deadline for the message. For applications such as news dissemination or advertising publishing, the message is only valid before its deadline. When the message reaches the deadline, there is no need to continue buffering the message. 

\subsection{Basic Definitions}

As a large number of mobile devices are carried by human-beings or relative to the society of people, we assume that the movement patterns of nodes are of periodicity. For example, Smith goes to work on Monday to Friday, so that his frequently visited locations might be composed of home, office, some bus-stations and a few nosheries. On the weekends, he usually goes to a club or some coffee bars, or other places different from those visited on workdays. However, his actions usually repeat per week, i.e., there exists some periodicity for his movement records. From this point, we assume the period time is $T$. Normally, the length of $T$ will be several days. Then we divide $T$ into $h$ slots. The length of each slot is thereby $\frac{T}{h}$. We represent the time points sequence as $<t_0,t_1, t_2, \ldots, t_h>$, then any two time points $t_s$ and $t_e$ $(t_s<t_e)$ form a time interval $[t_s,t_e]$, as shown in \figurename~\ref{fig:time_slots}. Thus we can extract the movement pattern of a node for the time interval $[t_s,t_e]$, by referring to the movement records of nodes. The definition of the record of a node is shown as follows.

\begin{figure}[!t]
\centering
  \includegraphics[width=0.75\linewidth]{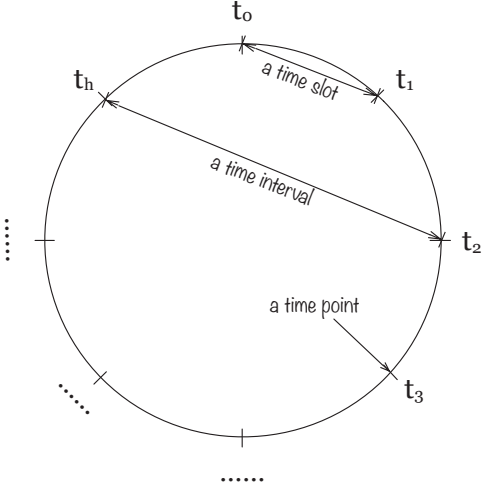}
  \caption{The period time $T$ is divided into h time slots.}
  \label{fig:time_slots}
\end{figure}

\begin{definition} Movement Record.\\
The movement record of a node $n_i$ is a $h\times m$ matrix $\mathbb{R}(i)$
\[
\mathbb{R}(i)=
\left.\left[
\overbrace{
\begin{array}{cccc}
\sfrac{1}{r_{1,1}^i} & \sfrac{1}{r_{1,2}^i} & \ldots & \sfrac{1}{r_{1,m}^i}\\
\sfrac{1}{r_{2,1}^i} & \sfrac{1}{r_{2,2}^i} & \ldots & \sfrac{1}{r_{2,m}^i} \\
\vdots & \vdots & \ddots & \vdots \\
\sfrac{1}{r_{h,1}^i} & \sfrac{1}{r_{h,2}^i} & \ldots & \sfrac{1}{r_{h,m}^i} \\
\end{array}}^{m~locations}
\right]\right\}h~time~slots
\]

of which the $k_{th}$ row $\mathbb{R}(i,k)$ is the movement record during the time slot $[t_{k-1},t_{k}]$, and we have \[\mathbb{R}(i,k)=[\sfrac{1}{r^i_{k,1}},\sfrac{1}{r^i_{k,2}},\ldots,\sfrac{1}{r^i_{k,m}}]\] where $r^i_{k,j}$ is the node $n_i$'s time interval on average to visit the location $a_j$, from the start to the end of the time slot $t_k$.
\end{definition}

From the above definition, we know that $\sfrac{1}{r_{k,j}^i}$ represents $n_i$'s average meeting frequency on the location $a_j$. Specifically, when $n_i$ has never arrived location $a_j$ during the whole $k_{th}$ time slot, we have $r_{k,j}^i=\infty$ and thus $\sfrac{1}{r_{k,1}^i}=0$. We can also get the average meeting time interval $M_{i,j}$ of $n_i$ on the location $a_j$ by averaging $r^{i}_{k,j}$ for all $k\in[1,h]$ in column $j$, formally we have
\begin{equation}
M_{i,j}=\left.\sum_{k=1}^{k=h}r^{i}_{k,j}\middle/h\right.
\label{eq:meeting_interval}
\end{equation}

\section{Overview of the Routing Problem}
\label{sec:overview}
Before discussing the details of our routing scheme, let us overview the routing problem. We first discuss how to extract the movement pattern of a set of node(s) in Section \ref{sec:Movement Pattern}. Then we analyze some key properties behind the routing in Section \ref{sec:Analysis}. Finally, the routing problem is formally defined as an optimal search problem in section \ref{sec:formalization}.

\subsection{Movement Pattern}
\label{sec:Movement Pattern}
The movement pattern is extracted from the records of node(s). We first define the function $\mathbbm{E}$ that transforms a movement record $\mathbb{R}$ to the corresponding movement pattern.

\begin{definition} Function $\mathbbm{E}$.
\[\mathbbm{E}([x_1,x_2,\ldots,x_m])=[\varkappa _1,\varkappa _2,\ldots,\varkappa _m]\]\textnormal{where}
\begin{equation}
\varkappa _j=\left\{
\begin{array}{cl}
 1 &x_j\geq\frac{\updelta}{m}\sum_{i=1}^{i=m}x_i\\
 0 & otherwise
\end{array}
\right.
\label{eq:extract}
\end{equation}
\textnormal{where $x_i$ is the $i$th element of the vector processed by $\mathbbm{E}$ and $0<\updelta<1$ is a system parameter.}
\label{def:function}
\end{definition}

The basic utility of function $\mathbbm{E}$ is to filter out the rarely visited locations so as to leave the frequently visited ones shown up. Following the above definition directly, we can define the movement pattern.

\begin{definition} Movement Pattern.
A movement pattern for a given node set $N$ at the time interval $[t_p,t_q]$ is represented as $\mathcal{P}(V,[t_s,t_e])$, where we have: 
\begin{equation}
\mathcal{P}(N,[t_s,t_e])= \mathbbm{E}(\sum_{n_x \in N}\sum_{i=s}^{i=e}\mathbb{R}(x,t_i))
\label{eq:pattern}
\end{equation}
\end{definition}

Notice that in the above definition, the movement pattern $\mathcal{P}$ is extracted from the accumulation of several records $\mathbb{R}$ of all nodes in set $N$ with a specific constraint on the time interval $[t_s,t_e]$, by function $\mathbbm{E}$. 

For example, assume that we have a network case, that there are four nodes and two locations in total and $n_4$ is the destination node, and $T$ is divided into two time slots $t_1$ and $t_2$, as shown in \figurename~\ref{fig:example1}. 

\begin{figure}[hbt]
\centering
  \includegraphics[width=0.9\linewidth]{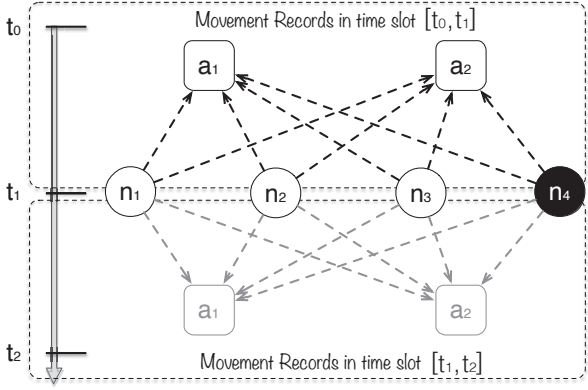}
  \caption{An example of the network instance}
  \label{fig:example1}
\end{figure}

\begin{table}[hbt]
\centering
  \begin{tabular}{|c|c|cc|}
  \hline
    & & $a_1$ & $a_2$  \\
    \hline
    
    \multicolumn{1}{|c|}{\multirow{4}{*}{$[t_0,t_1]$}} & $n_1$ &2.08 & 4.65 \\
    & $n_2$ & 4.4 & 4.3 \\
    & $n_3$ & 8.05 & 1.11 \\
    & $n_4$ & 2.6 & 3.3 \\
    \hline

    \multicolumn{1}{|c|}{\multirow{4}{*}{$[t_1,t_2]$}} & $n_1$ &6.02 & 2.95 \\
    & $n_2$ & 4.0 & 4.5 \\
    & $n_3$ & 6.05 & 15.49 \\
    & $n_4$ & 3.5 & 3.5 \\
    \hline
  \end{tabular}
  \caption{The weight of each edge in \figurename~\ref{fig:example1}}
  \label{tab:example1}
\end{table}

The weight of each edge has shown in \tablename~\ref{tab:example1}, which stands for the average visit time interval between a given "node-location" pair. Actually, these records are from the real experiment of four students' visiting interval (hours) to the two Computer Labs in Qingdao University. We had collected the data for about two weeks. Only the active 12 hours per day for students are counted. And we combine the 2 weeks datas as a one week data set such that the time period $T$ is set to $14\times 12=168$ hours. We set $[t_0,t_1]$ and $[t_1,t_2]$ to be the first and the last 84 hours, respectively. The record matrix for each student is as follow.
\[
\begin{array}{cc}
\mathbb{R}(1)=\left[
\begin{array}{cc}
\sfrac{1}{2.08} & \sfrac{1}{4.65} \\
\sfrac{1}{6.02} & \sfrac{1}{2.95}
\end{array}\right]&
\mathbb{R}(2)=\left[
\begin{array}{cc}
\sfrac{1}{4.4} & \sfrac{1}{4.3} \\
\sfrac{1}{4.0} & \sfrac{1}{4.5}
\end{array}\right]    \\ \\
\mathbb{R}(3)=\left[
\begin{array}{cc}
\sfrac{1}{8.05} & \sfrac{1}{1.11} \\
\sfrac{1}{6.05} & \sfrac{1}{15.49}
\end{array}\right]
&
\mathbb{R}(4)=\left[
\begin{array}{cc}
\sfrac{1}{2.6} & \sfrac{1}{3.3} \\
\sfrac{1}{3.5} & \sfrac{1}{3.5}
\end{array}\right]
\end{array}
\]


The movement pattern of the node set $\{n_1,n_2,n_3\}$ can be extracted from the movement records by equation~\ref{eq:pattern}. The accumulation of all the records is calculated as follows:
\begin{multline*}
\sum_{i=1}^{i=3}\sum_{j=1}^{j=2}\mathbb{R}(i,j) = \\\left[\frac{1}{2.08}+\frac{1}{4.4}+\frac{1}{8.05}+\frac{1}{6.02}+\frac{1}{4.0}+\frac{1}{6.05},\right.\\
\left.\frac{1}{4.65}+\frac{1}{4.3}+\frac{1}{1.11}+\frac{1}{2.95}+\frac{1}{4.5}+\frac{1}{15.49}\right]\\= [1.414,1.974]
\end{multline*}
We set $\updelta=0.95$ in equation~\ref{eq:extract}, and according to definition~\ref{def:function} we get
\[
\frac{\updelta}{m}\sum_{i=1}^{i=m}v_i=\frac{0.95}{2}\times(1.414+1.974)\approx 1.609
\] 
which is more than $1.414$ and is less than $1.974$, so that we have
\[
\mathcal{P}(\{n_1,n_2,n_3\},[t_0,t_2])=[0,1] 
\]
The movement pattern $\mathcal{P}$ indicates that the frequently visited location(s) set for the node set $\{n_1,n_2,n_3\}$ during the time interval $[t_0,t_2]$ is $A(\{n_1,n_2,n_3\},[t_0,t_2])=\{a_2\}$.
Note that there are totally $2^{(n-1)}-1$ non-empty subsets of $\overline{N}-{n_d}$, and each can be viewed as a entirety with some internal cooperation to delivery the message to $n_d$. We extract all the movement patterns for all non-empty subsets of $\{n_1,n_2,n_3\}$ and that of the node $n_4$. The results are compiled as follows.

\resizebox{\linewidth}{!}{
\begin{tabular}{ll}
 & \\
$\mathcal{P}(\{n_1\},[t_0,t_2])=[1,0]$ & $\mathcal{P}(\{n_1,n_2\},[t_0,t_2])=[1,0]$ \\
$\mathcal{P}(\{n_2\},[t_0,t_2])=[1,1]$ & $\mathcal{P}(\{n_1,n_3\},[t_0,t_2])=[0,1]$ \\
$\mathcal{P}(\{n_3\},[t_0,t_2])=[0,1]$ & $\mathcal{P}(\{n_2,n_3\},[t_0,t_2])=[0,1]$ \\
$\mathcal{P}(\{n_4\},[t_0,t_2])=[1,1]$ & $\mathcal{P}(\{n_1,n_2,n_3\},[t_0,t_2])=[0,1]$ \\
 &
\end{tabular}
}

Respectively, the corresponding sets of locations are listed as following.

\resizebox{\linewidth}{!}{
\begin{tabular}{ll}
 & \\
$A(\{n_1\},[t_0,t_2])=\{a_1\}$ & $A(\{n_1,n_2\},[t_0,t_2])=\{a_1\}$ \\
$A(\{n_2\},[t_0,t_2])=\{a_1,a_2\}$ & $A(\{n_1,n_3\},[t_0,t_2])=\{a_2\}$ \\
$A(\{n_3\},[t_0,t_2])=\{a_2\}$ & $A(\{n_2,n_3\},[t_0,t_2])=\{a_2\}$ \\
$A(\{n_4\},[t_0,t_2])=\{a_2\}$ & $A(\{n_1,n_2,n_3\},[t_0,t_2])=\{a_1,a_2\}$ \\
 &
\end{tabular}
}

\subsection{Analysis of Two Key Properties}
\label{sec:Analysis}
\subsubsection{Delivery Probability}
Assume that the movement pattern of the node(s) set $N$ during $[t_0,t_0+\tau_l]$ is $\mathcal{P}(N,[t_0,t_0+\tau_l])$ and the movement pattern of the destination node $n_d$ is $\mathcal{P}(n_d,[t_0,t_0+\tau_l])$, then we denote the common visited location(s) set by $A$, where we have $A=A(N,[t_0,t_0+\tau_l])\bigcap A((n_d,[t_0,t_0+\tau_l]))$. 
The basic idea is to view the relay node(s) set $N$ as an entirety to cooperatively deliver the message.
And the message can be delivered to the destination node $n_d$ by the node set $N$ if there exists an location $a_j$, where $\exists n_i\in N$ and $n_d$ encounter. In addition, if we denote the arrive time of $n_i$ and $n_d$ by the random variables $T_{i,j}$ and $T_{d,j}$ respectively, then they must satisfy $T_{i,j}<T_{d,j}<\tau_l$, i.e. $n_i$ comes to the location $a_j$ before $n_d$ and not exceeding the deadline $\tau_l$. For convenience, we let the current time point to be $t_0=0$. 
Besides, as stated in section~\ref{sec:preliminaries}, the time interval that each node visits an location follows an exponential distribution. From equation~\ref{eq:meeting_interval}, we can calculate the exponential distribution parameter of that node $n_i$ visiting location $a_j$ by
\[
\lambda_{i,j} = \sfrac{1}{M_{i,j}}
\]

For the case that $A$ is empty set, we set the related delivery probability as zero, and the expected delay to reach other location as $\infty$. For the case that $A$ is non-empty set,  we can estimate the direct delivery probability of $n_i$ to the destination node $n_d$ by equation~\ref{eq:probability}. Then the message delivery probability for the set $N$ can be estimated by equation~\ref{eq:probability2}. 

\resizebox{0.9\linewidth}{!}{
\begin{tabular}{l}
\parbox{\linewidth}{
\begin{multline}
P_{i,d}=1-\prod_{a_j\in A}\left(1-P(T_{i,j}<T_{d,j}<\tau_l) \right)\\
=1-\prod_{a_j\in A}\left(1-\int _{0}^{\tau _l}\!\!\!\!\int _{t_{i,j}}^{\tau _l}f_{T_{i,j},T_{d,j}}\left(t_{i,j},t_{d,j}\right)dt_{d,j}dt_{i,j}\right)\\
=1-\prod_{a_j\in A}\left(1-\int _{0}^{\tau _l}\!\!\!\!\int _{t_{i,j}}^{\tau _l}f_{T_{i,j}}(t_{i,j})f_{T_{d,j}}(t_{d,j}) dt_{d,j}dt_{i,j}\right)\\
=1-\prod_{a_j\in A}\left(1-\int _{0}^{\tau _l}\!\!\!\!\int _{t_{i,j}}^{\tau _l}\lambda _{d,j} \lambda _{i,j} e^{-\left(\lambda _{d,j}t_{d,j}+\lambda _{i,j}t_{i,j}\right)}dt_{d,j}dt_{i,j}\right)\\
=1-\prod_{a_j\in A}\left(1-\frac{\lambda _{i,j} \left(1-e^{-\tau _l \left(\lambda _{d,j}+\lambda _{i,j}\right)}\right)}{\lambda _{d,j}+\lambda _{i,j}} \right.\\ 
-\left.\left(e^{\tau _l \lambda _{i,j}}-1\right) e^{-\tau _l \left(\lambda _{d,j}+\lambda _{i,j}\right)}\right)
\label{eq:probability}
\end{multline}
}
\end{tabular}
}
\resizebox{0.9\linewidth}{!}{
\begin{tabular}{l}
\parbox{\linewidth}{
\begin{multline}
P_{N,d}=1-\prod_{n_i\in N}(1-P_{i,d}) \\
=1-\prod_{n_i\in N}\prod_{a_j\in A}\left(1-\frac{\lambda _{i,j} \left(1-e^{-\tau _l \left(\lambda _{d,j}+\lambda _{i,j}\right)}\right)}{\lambda _{d,j}+\lambda _{i,j}}\right. \\
-\left.\left(e^{\tau _l \lambda _{i,j}}-1\right) e^{-\tau _l \left(\lambda _{d,j}+\lambda _{i,j}\right)}\right)
\label{eq:probability2}
\end{multline}
}
\end{tabular}
}

Specially, when we set the time-to-live value for each message to be infinity, i.e. $\tau _l = \infty$, we have
\resizebox{0.9\linewidth}{!}{
\begin{tabular}{l}
\parbox{\linewidth}{
\begin{multline}
P_{i,d} = 1-\\
\prod_{a_j\in A}\left(1-\int _0^{\infty }\!\!\!\!\int _{t_{i,j}}^{\infty }\lambda _{d,j} \lambda _{i,j} e^{-\left(\lambda _{d,j} t_{d,j}+\lambda _{i,j} t_{i,j}\right)}dt_{d,j}dt_{i,j} \right)\\
=1-\prod_{a_j\in A}\left(1-\frac{\lambda _{i,j}}{\lambda _{d,j}+\lambda _{i,j}}\right)
\label{eq:probability_infty}
\end{multline}
}
\end{tabular}
}

and for $P_{N,d}$ we have

\resizebox{0.9\linewidth}{!}{
\begin{tabular}{l}
\parbox{\linewidth}{
\begin{multline}
P_{N,d} = 1-\prod_{n_i\in N}\prod_{a_j\in A}\left(1-\frac{\lambda _{i,j}}{\lambda _{d,j}+\lambda _{i,j}}\right) \\ 
=1-\prod_{n_i\in N}\prod_{a_j\in A}\frac{\lambda _{d,j}}{\lambda _{d,j}+\lambda _{i,j}}
\label{eq:probability2_infty}
\end{multline}
}
\end{tabular}
}

We assume that the message time-to-live $\tau _l=\infty$ and the destination node is $n_4$ in the network instance shown in \figurename~\ref{fig:example1}. The frequently visited locations for $N=\{n_1,n_2,n_3\}$ is $A(N,[t_0,t_3])=\{a_2\}$ and that for $n_4$ is $A(n_4,[t_0,t_3])=\{a_1,a_2\}$, so that we get the commonly visited location(s) set $A=\{a_2\}$.
The average meeting interval time for each node-location pair can be obtained by equation~\ref{eq:meeting_interval}.
\resizebox{\linewidth}{!}{
\begin{tabular}{l}
\parbox{\linewidth}{
\[
\begin{array}{cccc}
M_{1,1}=4.05 & M_{2,1}=4.2 & M_{3,1}=7.05 & M_{4,1}=3.05 \\
M_{1,2}=3.8 & M_{2,2}=4.4 & M_{3,2}=8.3 & M_{4,2}=3.4
\end{array}
\]
}
\end{tabular}
}

And then we can estimate the delivery probability for the node set $N$ as follows:

\resizebox{0.9\linewidth}{!}{
\begin{tabular}{l}
\parbox{\linewidth}{
\begin{multline*}
P_{\{n_1,n_2,n_3\},4}=\\
1-\left(\frac{\lambda_{4,2}}{\lambda_{1,2}+\lambda_{4,2}}\middle)\middle(\frac{\lambda_{4,2}}{\lambda_{2,2}+\lambda_{4,2}}\middle)\middle(\frac{\lambda_{4,2}}{\lambda_{3,2}+\lambda_{4,2}}\right) \\
= 1-\left(\frac{M_{1,2}}{M_{1,2}+M_{4,2}}\middle)\middle(\frac{M_{2,2}}{M_{2,2}+M_{4,2}}\middle)\middle(\frac{M_{3,2}}{M_{3,2}+M_{4,2}}\right) \\
= 0.789
\end{multline*}
}
\end{tabular}
}

\subsubsection{Expected Delay}

Besides the delivery probability, another property to explore is the node's expected delay before moving to other frequently visited locations. Assume that the frequent visited location(s) set of $n_i$ is $A_i$, then we let the random variable $D_i$ denote the minimum expected delay for $n_i$ to arrive at another location $a_j\in A_i$. Formally we have $D_i=\min\{T_{i,1},T_{i,2},\ldots,T_{i,|A_i|}\}$, and consequently we get the  probability density function of $D_i$ as

\[f_{D_i}(t)=\sum_{a_k\in A_i}\lambda_{i,k}\prod_{a_h\in A_i}e^{-\lambda_{i,h}t}\]

Thus we have

\begin{multline}
\label{eq:delay}
E[D_i]=\int_{-\infty}^{\infty}tf_{D_i}(t)dt \\
=\sum_{a_k\in A_i}\int_{0}^{\infty}t\lambda_{i,k}\prod_{a_h\in A_i}e^{-\lambda_{i,h}t}dt \\
=\sum_{a_k\in A_i}\frac{\lambda_{i,k}}{\left(\sum_{a_k\in A_i}\lambda_{i,h}\right)^2}\\
=\frac{1}{\sum_{a_k\in A_i}\lambda_{i,k}}\\
\end{multline}

For the example shown in \figurename~\ref{fig:example1}, we can calculate $E[D_i]$ for node $n_1$, $n_2$ and $n_3$ respectively according to equation~\ref{eq:delay}.
\[
\begin{array}{c}
E[D_1]=\sfrac{1}{\lambda_{1,1}}=M_{1,1}=4.05\\
E[D_2]=\frac{1}{\lambda_{2,1}+\lambda_{2,2}}=\frac{M_{2,1}M_{2,2}}{M_{2,1}+M_{2,2}}=2.15\\
E[D_3]=\sfrac{1}{\lambda_{3,2}}=M_{3,2}=8.3 
\end{array}
\]

For any pair of nodes $n_i$ and $n_j$, if $E[D_i]<E[D_j]$, then we can conclude that $n_i$ is expected to arrive at some new locations sooner, which means more connectivity opportunities to other nodes. In the intermittent network scenarios such as DTNs, when the multi-copy strategy is employed, it is vital to quickly spread the message over the network. Due to this reason, the expected delay $E[D]$ is a valuable metric for enhancing the routing performance. In Local-MPAR, we use $E[D]$ to evaluate the suitability of a node to be the active node, which is the only node that can change the state of the network for a certain message. In Tabu-MPAR, $E[D]$ is used as the utility to asymmetrically spraying a given number of replicas. In the example shown in \tablename~\ref{tab:example1}, node $n_2$ is the most suitable node to work as the active node, since it is expected to take the message to a new location more quickly than others.

\subsection{Problem Formalization}
\label{sec:formalization}
The routing objective in DTNs can be various. Basically, there are two primary metrics for routing performance evaluation, message delivery ratio and message average delay. Though the end-to-end delay should be viewed as a very important metric for optimization in some applications emphasizing the timeliness of messages, the routing performance is fundamentally based on an acceptable delivery ratio. From this point, a high delivery ratio is the guarantee of a good routing performance for any scenarios of DTNs. Based on this, the goal of our routing is to maximize the message delivery ratio for each generated message. Specifically, the goal is to select a subset $N\subseteq\overline{N}$ to make $P_{N,d}\geq P_{N',d}$ hold for $\forall N'\subseteq\overline{N}$, where $n_d$ is the destination node of the message.

For the example in \figurename~\ref{fig:example1}, we can get the estimated delivery probability of all non-empty subsets of $\{n_1,n_2,n_3\}$ as we do for it in section~\ref{sec:Analysis} by equation~\ref{eq:probability2_infty}:
\[
\begin{array}{ll}
P_{\{n_1\},4}=0.430 & P_{\{n_1,n_2\},4}=0.670  \\
P_{\{n_2\},4}=0.673 & P_{\{n_2,n_3\},4}=0.600  \\
P_{\{n_3\},4}=0.291 & P_{\{n_1,n_3\},4}=0.626  \\
P_{\{n_1,n_2,n_3\}}=0.789\\
\end{array}
\]

The optimal node(s) set is $\{n_1,n_2,n_3\}$ in this example. Notice that the estimated probability is not in direct proportion to the number of nodes holding the message. The delivery probability $P_{N,d}$ is focused on the node set $N$, which is viewed as an entirety during the routing process. It is possible to have that $P_{N,d}<P_{N',d}$ where $N'\subset N$. Actually, in the example shown in \figurename~\ref{fig:example1}, the delivery probability of the set $\{n_2,n_3\}$ is lower than that of the set $\{n_2\}$. If ignoring the constraint of buffer and energy resources, the routing strategy that broadcast each message over the network achieves the best performance, which means that a larger relay nodes may leads to a better (at least not worse) delivery performance. However, this ideal situation hardly exist in real network scenarios. So we tend to select the minimal node set that has a suitable movement pattern so as to maximize the estimated co-delivery probability for this node set. We can see from the example that if we treat the node set holding the same message as an entirety, then the estimated delivery probability is influenced by both the scale of the selected set and the movement pattern extracted from the records of this node set. Before formally defining our problem, firstly we give the definition of the optimal set $N_{opt}$.
\begin{definition} The optimal set $N_{opt}$\\
The optimal set $N_{opt}$ is a set of nodes in the network, where we have
\begin{displaymath}
N_{opt}= \min\{\mathrel{\mathop{\argmax}\limits_{N\subseteq \overline{N}}}P_{N,d}\}
\end{displaymath}
\end{definition}
The optimal set is the minimal node set that can achieve the highest delivery probability by cooperatively relay the message to the destination node, based on which we can formally define the problem.
\begin{definition}$N_{opt}$ Search Problem.\\
Let $f(\mathbf{x})$ to be the objective function and $g(\mathbf{x})$ the constraint function, where $\mathbf{x}=\{x_1,x_2,\ldots,x_n\}$ is the decision variable. And we have 

\resizebox{0.9\linewidth}{!}{
\begin{tabular}{l}
\parbox{\linewidth}{
\begin{multline*}
f(\mathbf{x})=\sum_{i=1}^{i=n}x_i\log\left(\prod_{a_j\in A}\left(1-\frac{\lambda _{i,j} \left(1-e^{-\tau _l \left(\lambda _{d,j}+\lambda _{i,j}\right)}\right)}{\lambda _{d,j}+\lambda _{i,j}}\right. \right.\\
\left.-\left.\left(e^{\tau _l \lambda _{i,j}}-1\right) e^{-\tau _l \left(\lambda _{d,j}+\lambda _{i,j}\right)}\right)\right)
\end{multline*}
}
\end{tabular}
}
\[
g(\mathbf{x})=\left|\sum_{i=1}^{i=n}x_i\right|
\]
Thus the problem can be formalized as
\[
\begin{array}{c}
\min f(\mathbf{x}) \\
s.t.\forall\mathbf{x'},~g(\mathbf{x'})-g(\mathbf{x})\geq 0, \\
\forall i,x_i\in \{0,1\}
\end{array} 
\]
And when the optimal solution $\mathbf{x}_{opt}$ is obtained, we have
\[
N_{opt}=\{n_i|x_i>0,1\leq i\leq n\}
\]
\label{def:N_opt}
\end{definition}

Notice that in the above definition, there exists a one-to-one match between $\mathbf{x}$ and $N$. Actually we have $f(\mathbf{x})=log(1-P_{N,d})$, where $f(\mathbf{x})$ is negative correlation with $P_{N,d}$. In this way, the $N_{opt}$ Search Problem is formalized as an optimization problem.

\section{Analysis for $N_{opt}$ Search Problem}
\label{sec:analysis}
In this section, we discuss the computational hardness of this problem and propose a heuristic method to approximately obtain the optimal solution. We first show that a newly added element in the input sequence can widely change the value of the objective function. For the offline algorithm, we put the analysis of the $NP-Hardness$. Finally we propose our algorithm based on the tabu search.

\subsection{Computational Hardness}
An online algorithm is one that can process its inputs piece-by-piece, i.e. it processes the input data in the sequence that they are fed to the algorithm, instead of requiring the entire input data set at the beginning. The input of our online algorithm is the movement records of all nodes in the network. We simply represent the record sequence as $\mathbb{R}_{1},\mathbb{R}_{2},\ldots,\mathbb{R}_{q}$ (each element in the sequence is a vector, e.g. $\mathbb{R}_j=\mathbb{R}(x,t_i)$ for node $n_x$ in the time slot $t_i$) and we currently only know the $1$---$k$ part of the sequence. Then what interests us is that how much the newly added element $\mathbb{R}_{k+1}$ can affect the solution, which is answered by Theorem~\ref{thm:online}. 

\begin{theorem}
For $\forall k\in[1,z)$, even the $1$---$k$ part of the sequence is known, the newly added element $\mathbb{R}_{k+1}$ can still change the movement pattern to any state with at least 1 non-zero element.
\label{thm:online}
\end{theorem}
\begin{proof}
See Appendix~A.
\end{proof}

Then we show that even all movement records are available for any node, there is not an efficient polynomial algorithm yet for solving the $N_{opt}$ problem.

\begin{theorem}
Even the global information is known in advance, i.e., the value of $\lambda_{i,j}$ is available for any node $n_i$, the $N_{opt}$ Search Problem is still NP-Hard. 
\label{thm:npc}
\end{theorem}
\begin{proof}
See Appendix~B.
\end{proof}

\subsection{The Trap of Local Optimum}
\label{sec:The Trap of Local Optimum}
A direct approach is to employ a local search algorithm to find the solution. However, this may lead us to the trap of local optimum. For the example shown in \figurename~\ref{fig:example1}, assume that node $n_2$ generates a message for $n_4$, then \tablename~\ref{tab:local_search} shows the calculation process of a local search algorithm.
\begin{table}[hbt]
  \caption{Local search process for the example in \figurename~\ref{fig:example1}}
  \centering
  \resizebox{\linewidth}{!}{
  \begin{tabular}{|c|c|c|c|c|c|c|}
    \hline
    & \multicolumn{3}{c|}{current solution}
    & \multicolumn{3}{c|}{next options}
    \\    
    \cline{2-7}
    
    & $\textbf{x}^{now}$ & N & $P_{N,4}$ & $\textbf{x}^{next}$ & N' & $P_{N',4}$ \\
    \hline
    
    \multicolumn{1}{|c|}{\multirow{3}{*}{\textbf{Step 1}}}
    &\multicolumn{1}{c|}{\multirow{3}{*}{$[0,1,0]$}}
    &\multicolumn{1}{c|}{\multirow{3}{*}{$\{n_2\}$}}
    &\multicolumn{1}{c|}{\multirow{3}{*}{$0.673$}}
    &$[0,1,\underline{1}]$ &$\{n_2,n_3\}$ & $0.600$ 
    \\
    
    \multicolumn{1}{|c|}{}
    &\multicolumn{1}{c|}{}
    &\multicolumn{1}{c|}{}
     &\multicolumn{1}{c|}{}
    &$[0,\underline{0},0]$ & $\phi$ & $0$
    \\
    
    \multicolumn{1}{|c|}{}
    &\multicolumn{1}{c|}{}
    &\multicolumn{1}{c|}{}
        &\multicolumn{1}{c|}{}
    &$[\underline{1},1,0]$ & $\{n_1,n_2\}$ & $0.670$
    \\
    
    \hline
    \multicolumn{1}{|c|}{\multirow{3}{*}{\textbf{Stop}}}
    &\multicolumn{1}{c|}{\multirow{3}{*}{$[0,1,0]$}}
    &\multicolumn{1}{c|}{\multirow{3}{*}{$\{n_2\}$}}
    &\multicolumn{1}{c|}{\multirow{3}{*}{$0.673$}}
    &\multicolumn{1}{c|}{\multirow{3}{*}{---}} 
    &\multicolumn{1}{c|}{\multirow{3}{*}{---}}  
    &\multicolumn{1}{c|}{\multirow{3}{*}{---}}   
    \\
    
    \multicolumn{1}{|c|}{}
    &\multicolumn{1}{c|}{}
    &\multicolumn{1}{c|}{}
     &\multicolumn{1}{c|}{}
    & &  & 
    \\
    
    \multicolumn{1}{|c|}{}
    &\multicolumn{1}{c|}{}
    &\multicolumn{1}{c|}{}
        &\multicolumn{1}{c|}{}
    & &  & 
    \\
    \hline    
  \end{tabular} }
  \label{tab:local_search}
\end{table}

In \tablename~\ref{tab:local_search}, only the node $n_2$ holds the message at the beginning, and the delivery probability of the set $N=\{n_2\}$ is $0.673$, which is higher than any of the next possible choice. So the algorithm mistakes $\textbf{x}^{now}$ as the optimal solution, while the real optimal one is $[1,1,1]$ that has the probability $0.789$. In this case, we have fallen into the trap of local optimum, and we cannot jump out of it by the local search algorithm.

\subsection{Heuristic Search in a Hypercube Space}
Theorem \ref{thm:npc} shows that there would not be any efficient polynomial algorithm for solving the $N_{opt}$ Search Problem and section~\ref{sec:The Trap of Local Optimum} shows that the local search algorithm may lead us to the trap of local optimum, which drive us to resort to some intelligent search schemes. In this paper, we leverage the tabu search algorithm to deal with the $N_{opt}$ problem. First of all we show the framework of the tabu search method in Algorithm~\ref{alg:framework}.

The tabu search process can be mainly divided into two steps, both of which require several basic rules and data structures shown in the notation table in algorithm~\ref{alg:framework}. The formal definitions of these notations will come up in the following context. The whole search process is to change the feasible solution step by step in the solution space, meanwhile evaluating the value of function $p$ and trying to obtain the optimal one. Step 1 is the initial step of the tabu search, where we choose an initial feasible solution in the solution space $S$ and represent it as $\textbf{x}^{now}$. Then we set the current best solution $\textbf{x}^{best}$ to be $\textbf{x}^{now}$. At this moment, there is no tabu records in the tabu table $\mathcal{T}$, as shown in line 3 of step 1. 

In step 2, if the stop rule $\mathcal{E}$ is satisfied, then the search procedure ends and returns the current best solution $\textbf{x}^{best}$ as the final result. Otherwise, we choose a subset from $\textbf{x}^{now}$'s neighborhoods $\mathcal{N}(\textbf{x}^{now})$ by referring to the current tabu table $\mathcal{T}$ and the pre-defined aspiration criteria $\mathcal{A}$. Then the algorithm search in the subset C and choose the solution $textbf{x}$ with maximal value of function $p$. In each iteration of step 2, we update the best record solution $\textbf{x}^{best}$. Finally we update the corresponding tabu element $\varphi$ in the tabu table $\mathcal{T}$ using the pre-assigned tabu length $\mathcal{L}(\varphi)$.

\begin{algorithm}[!tbp] 
\renewcommand{\algorithmicensure}{\textbf{Step}}
\caption{Tabu Search Framework.} 
\label{alg:framework} 
\begin{algorithmic}[1] 
\REQUIRE ~\\
\begin{center}
\resizebox{0.8\linewidth}{!}{
\begin{tabular}{|c|c|}
\hline
\textbf{notation} & \textbf{meaning}\\
\hline
$p$ & Evaluation function\\
$\mathcal{N}$ & Neighborhood of a solution\\
$\mathcal{S}$ & Solution space\\
$\mathcal{E}$ & Stop rules \\
$\mathcal{T}$ & tabu table $\left\{\begin{tabular}{c|c} 
               $\varphi$ & Tabu element \\
               \hline
               $\mathcal{L}$ & Tabu length\\
                \end{tabular}\right.$\\
$\mathcal{A}$ & Aspiration criteria\\
\hline
\end{tabular}}
\end{center}
\end{algorithmic} 
\begin{algorithmic}[1] 
\ENSURE \textbf{1:}
\STATE choose an initial feasible solution $\textbf{x}^{now}$ in $\mathcal{S}$
\STATE set the current best solution $\textbf{x}^{best}\leftarrow\textbf{x}^{now}$
\STATE set the tabu table $\mathcal{T}\leftarrow\emptyset$
\end{algorithmic} 
\begin{algorithmic}[1] 
\ENSURE \textbf{2:}
\IF{$\mathcal{E}(\textbf{x}^{best})=\textbf{true}$}
    \RETURN $\textbf{x}^{best}$
\ENDIF
\STATE get $C$ by filtering $\mathcal{N}(\textbf{x}^{now})$ according to $\mathcal{T}$ and $\mathcal{A}$
\STATE $\textbf{x}^{now}\leftarrow\mathrel{\mathop{\argmax}\limits_{\textbf{x}\in C}}p(\textbf{x})$
\IF{$p(\textbf{x}^{now})<p(\textbf{x}^{best})$}
    \STATE $\textbf{x}^{best}\leftarrow\textbf{x}^{now}$
\ENDIF
\STATE update $\mathcal{T}$
\STATE \textbf{goto} Step 2
\end{algorithmic}
\end{algorithm}

In the following context, we give the definitions of all the notations listed in the table in algorithm~\ref{alg:framework}.
\begin{definition}Evaluation function.\\
The evaluation function of the tabu search is represented as $p(\textbf{x})$, where
\[
p(\textbf{x})=-f(\textbf{x})
\]
\label{def:eva_func}
\end{definition}
In definition~\ref{def:eva_func}, we set the evaluation function by taking the opposite value of function $f(\textbf{x})$ in definition~\ref{def:N_opt}. Then the goal of our tabu search is to obtain the maximum value of the evaluation function. We give the definition about the neighborhood of a feasible solution in the tabu search in definition~\ref{def:neighborhood}.

\begin{definition}Neighborhood.\\
Assume that $\textbf{x}=[x_1,x_2,\ldots,x_n]$$(\forall i, x_i\in{0,1})$ is a feasible solution, then the neighborhood of $\textbf{x}$ in the tabu search is defined as follows
\begin{equation}
\mathcal{N}(\textbf{x})=\left\{~\textbf{x'}~\middle|~~\sum_{i=1}^{n}\left| x_i- x'_j\right|\leq 1\right\}
\end{equation}
\label{def:neighborhood}
\end{definition}
Definition~\ref{def:neighborhood} indicates that there is only one different element in a neighborhood solution compared to the original one. For example, if the current feasible solution $\textbf{x}$ is $[0,1,0,1]$, then its neighborhood $\mathcal{N}(\textbf{x})=\{[\underline{1},1,0,1],[0,\underline{0},0,1],[0,1,\underline{1},1],[0,1,0,\underline{0}]\}$, where each of the three vectors is a neighborhood solution for the current feasible solution $\textbf{x}$ with the underline element performed ``a bitwise NOT''. 

%

The stop rule is one of the most important rules for the tabu search, which determines when the tabu search procedure stops. The stop rule for our tabu search scheme is stated in definition~\ref{def:stop_rule}.

\begin{definition}Stop rule.\\
The stop rule $\mathcal{E}$ is defined as following
\[
\mathcal{E}(\textbf{x}^{best})=\left\{
\begin{array}{cl}
true & \textnormal{if $p(\textbf{x}^{best})$ not improved after $\theta$ steps} \\
false & \textnormal{otherwise}
\end{array}\right.
\]
\label{def:stop_rule}
\end{definition}

As stated before, the value of evaluation function $p$ is updated in every iteration of the search process. If the current $\textbf{x}^{now}$ is corresponding to a better function value $p(\textbf{x}^{now})>p(\textbf{x}^{best})$, then we update $\textbf{x}^{best}$ to $\textbf{x}^{now}$. However, if the best record value of the evaluation function $p(\textbf{x}^{best})$ has not been improved after quite a few steps, then it might be difficult to get a better solution without modifying the search algorithm itself. In this case, the search procedure stops and returns the current best record solution as the (approximate) optimal one.

The tabu table $\mathcal{T}$ is composed of two basic concepts, the tabu element $\varphi$ and its tabu length $\mathcal{L}(\varphi)$. The tabu element is the object to tabu in the search algorithm, which can be a simple variation on the solution or the variation of the evaluation function value. The tabu element in our tabu search scheme is defined in definition~\ref{def:tabu_element}.

\begin{definition}Tabu element.\\
A tabu element is a variation of the current solution vector, i.e. can be formally denoted by
\[
\varphi:\textbf{x}\rightarrow\textbf{x'}
\]
where we have
\[
\textbf{x}=(x_1,\ldots,x_{i-1},x_i,x_{i+1},\ldots,x_n)~~(i\in [0,n])
\]
and
\[
\textbf{x'}=(x_1,\ldots,x_{i-1},y_i,x_{i+1},\ldots,x_n)~~(i\in [0,n], x_i\neq y_i)
\]
\label{def:tabu_element}
\end{definition}
In fact, this tabu element is a movement from the current vertex to any adjacent vertex in the solution space, as shown in \figurename~\ref{fig:hypercube}, i.e. any movement from the current blue vertex to a red vertex can be viewed as a tabu element. Then we can define the tabu length $\mathcal{L}(\varphi)$ for each tabu element $\varphi$ in definition~\ref{def:tabu_length}.

\begin{definition}Tabu length.\\
The tabu length is denoted by $\mathcal{L}(\varphi)$, and we have
\[
\mathcal{L}(\varphi)=\lfloor T\rfloor
\]
where T is a random variable and $T\sim N(\mu,\sigma^2)$, and we set
\[
\mu=\left\{
\begin{array}{ll}
\sqrt{n}[1+p(\textbf{x'})-p(\textbf{x})] & p(\textbf{x'})>p(\textbf{x}) \\
\sqrt{n} & \textnormal{otherwise}
\end{array}
\right.
\]
\label{def:tabu_length}
\end{definition}

In definition~\ref{def:tabu_length}, the tabu length is set to be an rounded down random variable obeying the normal distribution, where the mean value $\mu$ guarantees that the variation of solution leading to a higher evaluation function value statistically has a longer tabu length.
The reason why we set the tabu length to be a random number instead of a constant one is to add some stochastic factors to our method, like what has been done in Sherwood Algorithm. The stochastic law can help us to jump out of some endless cycle of solutions, or to deal with some deliberate inputs. We establish the relationship between the mean value of $T$ and the variation of the evaluation function value $|p(\textbf{x})-p(\textbf{x'})|$. There are two different cases for the variation of the evaluation function value, $p(\textbf{x'})<p(\textbf{x})$ that indicates the value might have reached a new ``valley'' and $p(\textbf{x'})>p(\textbf{x})$ that indicates the value might have climbed on a higher mountain top. In the former case, we need a relatively large tabu length so as to jump out of the possible trap of local optimum. While in the later case, we should keep the tabu length smaller to avoid the value dropped to a deep valley again. Apparently, we have $\mu\in[\sqrt{n},2\sqrt{n}]$.

From definition \ref{def:tabu_element} and \ref{def:tabu_length}, the tabu table can be consequently defined in definition~\ref{def:tabu_table}.

\begin{definition}Tabu table.\\
The structure of the tabu table is as follows, \\
\begin{center}
$\mathcal{T}$=
  \begin{tabular}{|c|c|cc|c|c|}
  \hline
   $1$ & $2$ & $\cdots$ & $\cdots$ & $n-1$ & $n$\\
   \hline
   $t_1$ & $t_2$ & $\cdots$ & $\cdots$ & $t_{n-1}$ & $t_n$ \\
    \hline
  \end{tabular} \\
\end{center}
Assume that the tabu element is $\varphi:\textbf{x}->\textbf{x'}$, the update rule is defined as 
\[
\forall i, t_i=\left\{
\begin{array}{ll}
\mathcal{L}(\varphi) & \textnormal{if $|x_i-x'_i|=1$} \\
  0 & \textnormal{if}~t_i=0 \\
t_i-1 &  \textnormal{otherwise} \\
\end{array}
\right.
\]
\label{def:tabu_table}
\end{definition}
The tabu table $\mathcal{T}$ is a $2\times n$ table, where the first row stands for the variation position of the solution vector, and the second row indicates the corresponding tabu length for each variation position. For example, assume that the variation of solution $\varphi:[0,1,0]\rightarrow[1,1,0]$ happens, then we should update the current tabu table ($t_1,t_2,t_3 > 0$) as
\begin{center}
\begin{tabular}{|c|c|c|}
\hline
1 & 2 & 3 \\
\hline
$t_1$ & $t_2$ & $t_3$ \\
\hline
\end{tabular} $\longrightarrow$
\begin{tabular}{|c|c|c|}
\hline
1 & 2 & 3 \\
\hline
$\mathcal{L}(\varphi)$ & $t_2-1$ & $t_3-1$ \\
\hline
\end{tabular}
\end{center}
In each update operation, the corresponding variation position of the solution in this table is filled with the tabu length $\mathcal{L}(\varphi)$, and all non-zero lengths are reduced by 1. Any position in the solution vector with a non-zero length is tabued, i.e. can only be modified after passing its tabu length steps.
 
Another important rule for the tabu search algorithm is the aspiration criterion that indicates the absolvable tabu elements in the tabu table $\mathcal{T}$. The aspiration for our tabu search scheme is stated in definition~\ref{def:aspiration}
\begin{definition}Aspiration criterion.\\
The aspiration criterion is defined as following
\begin{center}
\begin{tabular}{|c|}
\hline
For $\forall$ solution $\textbf{x}$, if we have $p(\textbf{x})>p(\textbf{x}^{best})$ \\
then the solution $\textbf{x}$ can be chosen even if the \\
corresponding position is tabued in $\mathcal{T}$.\\
\hline
\end{tabular}
\end{center}
\label{def:aspiration}
\end{definition}

When all the positions are tabued in $\mathcal{T}$, the current solution vector can not move to any of its neighborhood. However if there is a neighborhood solution $\textbf{x}$ has the better evaluation function value than the current recorded best one, we absolve the corresponding tabu element and move the solution vector to $\textbf{x}$.

The tabu search process of the example in \figurename~\ref{fig:example1} is shown in \tablename~\ref{tab:tabu_search}. Like the local process illustrated in \tablename~\ref{tab:example1}, the source node is $n_2$ and thus the initial solution is $[0,1,0]$. For simplification, the tabu length $\mathcal{L}$ constants at 3 rather than being randomly chosen. Besides, the variable $\theta$ in definition~\ref{def:stop_rule} of stop rule is set to be 3.

In step 1, the tabu table $\mathcal{T}$ is empty, so all the next options are choosable. According to line 5 in algorithm~\ref{alg:framework}, we choose the solution vector $[1,1,0]$ that has the maximal evaluation function value. In step 2, the first position of the solution vector is tabued, so the next solution vector can be chosen only between $[0,0,0]$ and $[1,1,1]$. The search procedure continues until it reaches step 5, where the best recorded solution has not changed after $\theta=3$ steps, so the tabu search procedure stops. Overview the whole process, we can quickly go through the solution space and jump out of the trap of local optimum, which cannot be achieved by the local search algorithm in \tablename~\ref{tab:local_search}.

\begin{table*}[hbt]
  \caption{Tabu search process for the example in \figurename~\ref{fig:example1}}
  \resizebox{\textwidth}{!}{
  \begin{tabular}{|c|c|c|c|c|c|c|c|c|c|c|c|}
    \hline
    & \multicolumn{3}{c|}{current solution}
    & \multicolumn{3}{c|}{current best solution}
    & \multicolumn{1}{c|}{\multirow{2}{*}{tabu table $\mathcal{T}$}}
    & \multicolumn{4}{c|}{next options}
    \\    
    \cline{2-7}\cline{9-12}
    
    & $\textbf{x}^{now}$ & N & $P_{N,4}$ & $\textbf{x}^{best}$ & $N^{best}$ & $P_{N^{best},4}$ & & $\textbf{x}^{next}$ & N' & $P_{N',4}$ & status \\
    \hline
%
    \multicolumn{1}{|c|}{\multirow{3}{*}{\textbf{Step 1}}}
    &\multicolumn{1}{c|}{\multirow{3}{*}{$[0,1,0]$}}
    &\multicolumn{1}{c|}{\multirow{3}{*}{$\{n_2\}$}}
    &\multicolumn{1}{c|}{\multirow{3}{*}{$0.673$}}
    &\multicolumn{1}{c|}{\multirow{3}{*}{$[0,1,0]$}}
    &\multicolumn{1}{c|}{\multirow{3}{*}{$\{n_2\}$}}
    &\multicolumn{1}{c|}{\multirow{3}{*}{$0.673$}}
    &\multicolumn{1}{c|}{\multirow{3}{*}{$
                                                     \begin{tabular}{|c|c|c|}
                                                        \hline
                                                        1 & 2 & 3 \\
                                                        \hline
                                                        0 & 0 & 0 \\
                                                        \hline
                                                     \end{tabular}$}}
    &$[\underline{1},1,0]$ &$\{n_1,n_2\}$ & $0.670$ & choosable
    \\
    
    \multicolumn{1}{|c|}{}
    &\multicolumn{1}{c|}{}
    &\multicolumn{1}{c|}{}
    &\multicolumn{1}{c|}{}
    &\multicolumn{1}{c|}{}
    &\multicolumn{1}{c|}{}
    &\multicolumn{1}{c|}{}
    &\multicolumn{1}{c|}{}
    &$[0,\underline{0},0]$ & $\phi$ & $0$ & choosable
    \\
    
    \multicolumn{1}{|c|}{}
    &\multicolumn{1}{c|}{}
    &\multicolumn{1}{c|}{}
    &\multicolumn{1}{c|}{}
    &\multicolumn{1}{c|}{}
    &\multicolumn{1}{c|}{}
    &\multicolumn{1}{c|}{}
    &\multicolumn{1}{c|}{}
    &$[0,1,\underline{1}]$ & $\{n_2,n_3\}$ & $0.600$ & choosable
    \\
    \hline
    
    \multicolumn{1}{|c|}{\multirow{3}{*}{\textbf{Step 2}}}
    &\multicolumn{1}{c|}{\multirow{3}{*}{$[1,1,0]$}}
    &\multicolumn{1}{c|}{\multirow{3}{*}{$\{n_1,n_2\}$}}
    &\multicolumn{1}{c|}{\multirow{3}{*}{$0.670$}}
    &\multicolumn{1}{c|}{\multirow{3}{*}{$[0,1,0]$}}
    &\multicolumn{1}{c|}{\multirow{3}{*}{$\{n_2\}$}}
    &\multicolumn{1}{c|}{\multirow{3}{*}{$0.673$}}
    &\multicolumn{1}{c|}{\multirow{3}{*}{$
                                                     \begin{tabular}{|c|c|c|}
                                                        \hline
                                                        1 & 2 & 3 \\
                                                        \hline
                                                        3 & 0 & 0 \\
                                                        \hline
                                                     \end{tabular}$}}
    &$[\underline{0},1,0]$ &$\{n_2\}$ & $0.673$ & tabu
    \\
    
    \multicolumn{1}{|c|}{}
    &\multicolumn{1}{c|}{}
    &\multicolumn{1}{c|}{}
    &\multicolumn{1}{c|}{}
    &\multicolumn{1}{c|}{}
    &\multicolumn{1}{c|}{}
    &\multicolumn{1}{c|}{}
    &\multicolumn{1}{c|}{}
    &$[0,\underline{0},0]$ & $\phi$ & $0$ & choosable
    \\
    
    \multicolumn{1}{|c|}{}
    &\multicolumn{1}{c|}{}
    &\multicolumn{1}{c|}{}
    &\multicolumn{1}{c|}{}
    &\multicolumn{1}{c|}{}
    &\multicolumn{1}{c|}{}
    &\multicolumn{1}{c|}{}
    &\multicolumn{1}{c|}{}
    &$[1,1,\underline{1}]$ & $\{n_1,n_2,n_3\}$ & $0.789$ & choosable
    \\

    \hline
    \multicolumn{1}{|c|}{\multirow{3}{*}{\textbf{Step 3}}}
    &\multicolumn{1}{c|}{\multirow{3}{*}{$[1,1,1]$}}
    &\multicolumn{1}{c|}{\multirow{3}{*}{$\{n_1,n_2,n_3\}$}}
    &\multicolumn{1}{c|}{\multirow{3}{*}{$0.789$}}
    &\multicolumn{1}{c|}{\multirow{3}{*}{$[1,1,1]$}}
    &\multicolumn{1}{c|}{\multirow{3}{*}{$\{n_1,n_2,n_3\}$}}
    &\multicolumn{1}{c|}{\multirow{3}{*}{$0.789$}}
    &\multicolumn{1}{c|}{\multirow{3}{*}{$
                                                     \begin{tabular}{|c|c|c|}
                                                        \hline
                                                        1 & 2 & 3 \\
                                                        \hline
                                                        2 & 0 & 3 \\
                                                        \hline
                                                     \end{tabular}$}}
    &$[\underline{0},1,1]$ &$\{n_2,n_3\}$ & $0.600$ & tabu
    \\
    
    \multicolumn{1}{|c|}{}
    &\multicolumn{1}{c|}{}
    &\multicolumn{1}{c|}{}
    &\multicolumn{1}{c|}{}
    &\multicolumn{1}{c|}{}
    &\multicolumn{1}{c|}{}
    &\multicolumn{1}{c|}{}
    &\multicolumn{1}{c|}{}
    &$[1,\underline{0},1]$ & $\{n_1,n_3\}$ & $0.626$ & choosable
    \\
    
    \multicolumn{1}{|c|}{}
    &\multicolumn{1}{c|}{}
    &\multicolumn{1}{c|}{}
    &\multicolumn{1}{c|}{}
    &\multicolumn{1}{c|}{}
    &\multicolumn{1}{c|}{}
    &\multicolumn{1}{c|}{}
    &\multicolumn{1}{c|}{}
    &$[1,1,\underline{0}]$ & $\{n_1,n_2\}$ & $0.670$ & tabu
    \\

    \hline
    \multicolumn{1}{|c|}{\multirow{3}{*}{\textbf{Step 4}}}
    &\multicolumn{1}{c|}{\multirow{3}{*}{$[1,0,1]$}}
    &\multicolumn{1}{c|}{\multirow{3}{*}{$\{n_1,n_3\}$}}
    &\multicolumn{1}{c|}{\multirow{3}{*}{$0.626$}}
    &\multicolumn{1}{c|}{\multirow{3}{*}{$[1,1,1]$}}
    &\multicolumn{1}{c|}{\multirow{3}{*}{$\{n_1,n_2,n_3\}$}}
    &\multicolumn{1}{c|}{\multirow{3}{*}{$0.789$}}
    &\multicolumn{1}{c|}{\multirow{3}{*}{$
                                                     \begin{tabular}{|c|c|c|}
                                                        \hline
                                                        1 & 2 & 3 \\
                                                        \hline
                                                        1 & 3 & 2 \\
                                                        \hline
                                                     \end{tabular}$}}
    &$[\underline{0},0,1]$ &$\{n_3\}$ & $0.291$ & tabu
    \\
    
    \multicolumn{1}{|c|}{}
    &\multicolumn{1}{c|}{}
    &\multicolumn{1}{c|}{}
    &\multicolumn{1}{c|}{}
    &\multicolumn{1}{c|}{}
    &\multicolumn{1}{c|}{}
    &\multicolumn{1}{c|}{}
    &\multicolumn{1}{c|}{}
    &$[1,\underline{1},1]$ & $\{n_1,n_2,n_3\}$ & $0.789$ & tabu
    \\
    
    \multicolumn{1}{|c|}{}
    &\multicolumn{1}{c|}{}
    &\multicolumn{1}{c|}{}
    &\multicolumn{1}{c|}{}
    &\multicolumn{1}{c|}{}
    &\multicolumn{1}{c|}{}
    &\multicolumn{1}{c|}{}
    &\multicolumn{1}{c|}{}
    &$[1,0,\underline{0}]$ & $\{n_1\}$ & $0.430$ & tabu
    \\
    
    \hline
    \multicolumn{1}{|c|}{\multirow{3}{*}{\textbf{Step 5}}}
    &\multicolumn{1}{c|}{\multirow{3}{*}{$[1,0,1]$}}
    &\multicolumn{1}{c|}{\multirow{3}{*}{$\{n_1,n_3\}$}}
    &\multicolumn{1}{c|}{\multirow{3}{*}{$0.626$}}
    &\multicolumn{1}{c|}{\multirow{3}{*}{$[1,1,1]$}}
    &\multicolumn{1}{c|}{\multirow{3}{*}{$\{n_1,n_2,n_3\}$}}
    &\multicolumn{1}{c|}{\multirow{3}{*}{$0.789$}}
    &\multicolumn{1}{c|}{\multirow{3}{*}{$
                                                     \begin{tabular}{|c|c|c|}
                                                        \hline
                                                        1 & 2 & 3 \\
                                                        \hline
                                                        0 & 2 & 1 \\
                                                        \hline
                                                     \end{tabular}$}}
    &$[\underline{0},0,1]$ &$\{n_3\}$ & $0.291$ & choosable
    \\
    
    \multicolumn{1}{|c|}{}
    &\multicolumn{1}{c|}{}
    &\multicolumn{1}{c|}{}
    &\multicolumn{1}{c|}{}
    &\multicolumn{1}{c|}{}
    &\multicolumn{1}{c|}{}
    &\multicolumn{1}{c|}{}
    &\multicolumn{1}{c|}{}
    &$[1,\underline{1},1]$ & $\{n_1,n_2,n_3\}$ & $0.789$ & tabu
    \\
    
    \multicolumn{1}{|c|}{}
    &\multicolumn{1}{c|}{}
    &\multicolumn{1}{c|}{}
    &\multicolumn{1}{c|}{}
    &\multicolumn{1}{c|}{}
    &\multicolumn{1}{c|}{}
    &\multicolumn{1}{c|}{}
    &\multicolumn{1}{c|}{}
    &$[1,0,\underline{0}]$ & $\{n_1\}$ & $0.430$ & tabu
    \\
    
    \hline
    \multicolumn{1}{|c|}{\multirow{3}{*}{\textbf{Stop}}}
    &\multicolumn{1}{c|}{\multirow{3}{*}{$[1,1,1]$}}
    &\multicolumn{1}{c|}{\multirow{3}{*}{$\{n_1,n_2,n_3\}$}}
    &\multicolumn{1}{c|}{\multirow{3}{*}{$0.789$}}
    &\multicolumn{1}{c|}{\multirow{3}{*}{---}}
    &\multicolumn{1}{c|}{\multirow{3}{*}{---}}
    &\multicolumn{1}{c|}{\multirow{3}{*}{---}}
    &\multicolumn{1}{c|}{\multirow{3}{*}{---}}
    &\multicolumn{1}{c|}{\multirow{3}{*}{---}} 
    &\multicolumn{1}{c|}{\multirow{3}{*}{---}} 
    &\multicolumn{1}{c|}{\multirow{3}{*}{---}}  
    &\multicolumn{1}{c|}{\multirow{3}{*}{---}}
    \\
    
    \multicolumn{1}{|c|}{}
    &\multicolumn{1}{c|}{}
    &\multicolumn{1}{c|}{}
    &\multicolumn{1}{c|}{}
    &\multicolumn{1}{c|}{}
    &\multicolumn{1}{c|}{}
    &\multicolumn{1}{c|}{}
    &\multicolumn{1}{c|}{}
    & &  &  &
    \\
    
    \multicolumn{1}{|c|}{}
    &\multicolumn{1}{c|}{}
    &\multicolumn{1}{c|}{}
    &\multicolumn{1}{c|}{}
    &\multicolumn{1}{c|}{}
    &\multicolumn{1}{c|}{}
    &\multicolumn{1}{c|}{}
    &\multicolumn{1}{c|}{}
    & &  &  &
    \\    
    \hline    
  \end{tabular}}
  \label{tab:tabu_search}
\end{table*}

\section{Movement Pattern-Aware Optimal Routing}
\label{sec:routing}
In this section, we give the details about the routing process. We implement two routing protocols, Local-MPAR and Tabu-MPAR, based on two search algorithms, the local search algorithm and the tabu search algorithm respectively. In Local-MPAR, the other node(s)'s $\lambda$ value for any one node is required when they encounter, while in Tabu-MPAR, we assume that each node knows the $\lambda$ values of all other node(s) in advance.

In this paper, we investigate the routing process for a single message (it may have one or more than one replicas). To illustrate this further, we refer to the concept in Epidemic routing protocol \cite{Vahdat2000} that treats the message as virus. Then a successful message replication operation can be viewed as an infection process. The delivery operation is excluded in our discussion, i.e., when we say $n_a$ infects $n_b$, $n_b$ is not the destination node. Start from this point, we can classify all the nodes in the network into three kinds. 
\begin{itemize}
\item \textit{Infected node.} An infected node is a node holding the message or any of its replicas. However an infected node cannot infect other pure nodes.
\item \textit{Pure node.} A pure node is a node not holding the message or any of its replicas. It can be infected by any infectious node.
\item \textit{Infectious node.} An infectious node is a special kind of infected node that is able to infect other pure nodes.
\end{itemize}

An \textit{infectious} state is able to move to an \textit{infected} state, which means that it still holds the message but unable to add new ones to the network. A \textit{pure} state can also be transferred to an \textit{infected} state or infectious state, which means that it receives the message. We further explain these states by referring to Epidemic \cite{Vahdat2000} and SprayAndWait \cite{Spyropoulos2005} protocols.

 
In Epidemic routing, there are only two kinds of nodes, \textit{pure node} and \textit{infectious node}. The replicas distribution process ends when all the pure nodes are infected and become infectious nodes. SprayAndWait routing employs all the three kinds of nodes. In its source edition, there is only one infectious node in the network, while in its binary edition the infectious nodes can be multiple but limited in a fixed number. The replicas distribution process ends when there is no infectious node in the network. 

Actually, these three kinds of node exist in all kinds of routing protocols in mobile delay tolerant networks. Specifically, Epidemic and SprayAndWait are two most representative zero-knowledge based routing protocols. All other replicas distribution process in any routing protocol can more or less be viewed as the improvement of these two basic schemes, by letting the node state transfers among \textit{pure}, \textit{infected} and \textit{infectious}. Before introducing the Local-MPAR and Tabu-MPAR, we propose three postulations about the routing.

\begin{postulation}
\label{post:replication}
Any node $n_a$ replicates the message to $n_b$ only if $n_b$ currently does not hold any of its replicas.
\end{postulation}

\begin{postulation}
\label{post:deletion}
The operation of cleaning messages out of the buffer caused by time-to-live deadline or buffer size constraint is called ``drop'', while that caused by state transition is called ``deletion''.
\end{postulation}

\begin{postulation}
The set of nodes holding the message or any of its replicas is denoted by $N$, which is compose of all infectious nodes and infected nodes.
\end{postulation}

Then we have the following theorem.

\begin{theorem}
\label{theorem:replication_deletion}
The replication operation is corresponding to the transition from W to G and deletion operation to the transition from G to W.
\end{theorem}
\begin{proof}
From the definition of the three states, a node holding the message can be either in state B or G, and a node not holding the message is accurately corresponding to state W. Since that there is no direct transition between B and W, a node in state B transfers to W has to pass the state G, and vise versa. Then from postulation~1, the replication operation is corresponding to the transition from W to B or W to T, i.e. W to G. From postulation~2, the deletion operation is corresponding to the transition from B to W or G to W, i.e. G to W.
\end{proof}

\subsection{Local Search Based Movement Pattern-Aware Routing}
In Local-MPAR algorithm, each node could be of all the three states. However, we only allow not more than 1 \textit{infectious} node exists for each generated message. The basic idea of Local-MPAR is to dynamically adjust the set $N$ that includes all nodes holding the message, so as to maximize the co-delivery probability $P_{N,d}$.

Algorithm~\ref{alg:Local-MPAR} shows the Local-MPAR process. There are three stages in Local-MPAR. In the initial stage, the source node $n_s$ generates a message $M$, and then set its state to be infectious. Correspondingly, $N$ is initialized to be composed of only $n_s$.  For the node set $N$, we have the following postulation.
\begin{postulation}
\label{post:N}
In Local-MPAR, the node set $N$ is always maintained by the active node.
\end{postulation}
The whole routing process can be viewed as the dynamical adjustment of the node set $N$, as shown in the second stage. The key operation of this stage is to let each node finish the state transition according to the rules shown in\tablename~\ref{tab:Local-MPAR}. For simplicity, in the following context we name the table position in row B column G as BG, by parity of reasoning. We can see from \tablename~\ref{tab:Local-MPAR} that the state transition happens only if either $n_a$ or $n_b$ is in infectious state, i.e. state B. There is nothing in BB because that it cannot be true that both $n_a$ and $n_b$ are in infectious state in Local-MPAR. 

\begin{figure}[!t]
\centering
\includegraphics[width=2in]{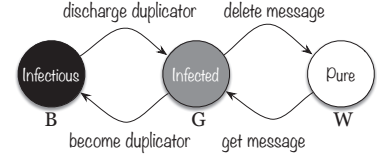}
\caption{Node state transition in Local-MPAR.}
\label{fig:states2}
\end{figure}

\begin{lemma}
\label{lemma:infectious}
The replication or deletion of message happens only if one of the encountered two nodes is in infectious state.
\end{lemma}
\begin{proof}
Directly obtained by \tablename~\ref{tab:Local-MPAR}.
\end{proof}

\begin{theorem} 
The routing stage of Local-MPAR is a local search process in the solution space $\mathcal{S}$.
\end{theorem}
\begin{proof}
From postulation~4, we know that the node set $N$ is updated only in infectious node. Lemma~1 directly shows that any happened replication or deletion operation can be informed to the infectious node immediately, so that $N$ would be updated in time. Either a replication or deletion operation just causes 1 replicas difference in the network, thus only adding or removing one element in the set $N$, which indeed transfers the correspondent solution vector to one of its neighborhoods $\mathcal{N}(N)$. We can see from \tablename~4 that N would not vary if and only if $\forall N'\in\mathcal{N}(N), P_{N',d}\leq P_{N,d}$, i.e.,  if and only if the solution reaches the local optimum.
\end{proof}

There still remains one question. Since there is only one infectious node in Local-MPAR, how should we choose the infectious node? The basic principle is to choose the node with the smallest expected delay to arrive at another frequently visited location, i.e. $E[D]$ in equation~\ref{eq:delay}. In \tablename~\ref{tab:Local-MPAR} GB, when there is no need to change $n_a$ from G to W, we should consider that whether to change $n_a$ to be the infectious node. If $E[D_a]$ is less than $E[D_b]$, then we deem that $n_a$ is more suitable to be the message duplicator, because it is expected to visit other location(s) sooner, which means more transmission opportunities. Conversely, if the duplicator responsibility moves to $n_b$, i.e. G$\rightarrow$B happens in $n_b$, we should symmetrically have B$\rightarrow$G happens in $n_a$ to keep the uniqueness property of the infectious node, as shown in \tablename~\ref{tab:Local-MPAR} BG.

The routing process of the network instance of \figurename~\ref{fig:example1} is shown in \figurename~\ref{fig:local_rt}. At time $t_1$, node $n_2$ generates the message of which the destination is $n_4$. From time $t_2$ to $t_3$, since that the co-delivery probability of both $\{n_1,n_2\}$ and $\{n_2,n_3\}$ is less than $\{n_2\}$. However, from the discussion before, we know that the set $N_{opt}=\{n_1,n_2,n_3\}$ could reach the best delivery ratio. However in Local-MPAR $N$ is not able to evolve to $N_{opt}$ due to the trap of local optimum, as same as in \tablename~\ref{tab:local_search}.

\begin{table*}[hbt]
\centering
  \caption{Node state transition rules for Local-MPAR}
  \label{tab:Local-MPAR}
  \resizebox{\textwidth}{!}{
  \begin{tabular}{|c|c|c|c|c|}
  \hline

   &\multicolumn{4}{c|}{\textbf{state of $n_{b}$}}
   \\
   \hline
    \multicolumn{1}{|c|}{\multirow{9}{*}{\textbf{state of $n_{a}$}}}
   &
   & \textbf{B}
   & \textbf{G}
   & \textbf{W}
   \\
   \cline{2-5}
   & \textbf{B}
   & \multicolumn{1}{c:}{---}
   & \multicolumn{1}{c:}{B
     $\rightarrow\left\{
     \begin{array}{ll} 
      \textnormal{G} & \textnormal{if}~\textnormal{G}\rightarrow\textnormal{B happens in }n_b\\
      \textnormal{B} & \textnormal{else}
     \end{array}\right.
     $}
   & B
   \\
   \cline{2-2}\cdashline{3-5}
   & \textbf{G}
   & \multicolumn{1}{c:}{G
     $\rightarrow\left\{
     \begin{array}{ll}
      \textnormal{W} & \textnormal{if}~P_{N-\{n_a\},d}>P_{N,d} \\
      \textnormal{B} & \textnormal{if}~P_{N-\{n_a\}}\leq P_{N,d}~\textnormal{and}~E[D_a]<E[D_b] \\
      \textnormal{G} & \textnormal{else}
     \end{array}\right.
     $}
   & \multicolumn{1}{c:}{G}
   & G
   \\
   \cline{2-2}\cdashline{3-5}
   & \textbf{W}
   & \multicolumn{1}{c:}{W
     $\rightarrow\left\{
     \begin{array}{ll}
      \textnormal{G} & \textnormal{if}~P_{N\bigcup\{n_b\},d}>P_{N,d} \\
      \textnormal{W} & \textnormal{else}
     \end{array}\right.
     $}
   & \multicolumn{1}{c:}{W}
   & W
   \\
  \hline
  \end{tabular}}
\end{table*}

\begin{table*}[hbt]
\centering
  \caption{Node state transition rules for Tabu-MPAR}
  \label{tab:Tabu-MPAR}
  \resizebox{\textwidth}{!}{
  \begin{tabular}{|c|c|c|c|c|}
  \hline

   &\multicolumn{4}{c|}{\textbf{state of $n_{b}$}}
   \\
   \hline
    \multicolumn{1}{|c|}{\multirow{9}{*}{\textbf{state of $n_{a}$}}}
   &
   & \textbf{B}
   & \textbf{G}
   & \textbf{W}
   \\
   \cline{2-5}
   & \textbf{B}
   & \multicolumn{1}{c:}{B}
   & \multicolumn{1}{c:}{
       \begin{tabular}{c}
     $B\rightarrow B|G$\\
     \textit{(depends on the tickets left in $n_a$} \\
     \textit{after allocating to $n_b$)} \\
     \end{tabular}
   }
   & 
    \begin{tabular}{c}
     $B\rightarrow B|G$\\
     \textit{(depends on the} \\
     \textit{tickets left in $n_a$} \\
     \textit{after allocating to $n_b$)} \\
     \end{tabular}
   \\
   \cline{2-2}\cdashline{3-5}
   & \textbf{G}
   & \multicolumn{1}{c:}{
   \begin{tabular}{c}
     $G\rightarrow B|G$\\
     \textit{(depends on the} \\
     \textit{tickets left in $n_a$} \\
     \textit{after allocating to $n_b$)} \\
     \end{tabular}
   }
   & \multicolumn{1}{c:}{G}
   & G
    $\rightarrow\left\{
     \begin{array}{ll}
      W & \textnormal{if}~n_a\notin N_{opt}~\textnormal{and}~n_b\in N_{opt}\\
      G & \textnormal{else}
     \end{array}\right.
    $
   \\
   \cline{2-2}\cdashline{3-5}
   & \textbf{W}
   & \multicolumn{1}{c:}{
     \begin{tabular}{c}
     $W\rightarrow B|G$\\
     \textit{(depends on the} \\
     \textit{tickets left in $n_a$} \\
     \textit{after allocating to $n_b$)} \\
     \end{tabular}
     }
   & \multicolumn{1}{c:}{W
    $\rightarrow\left\{
    \begin{array}{ll}
     \textnormal{G} & \textnormal{if}~n_a\in N_{opt}~\textnormal{and}~n_b\notin N_{opt}\\
     \textnormal{W} & \textnormal{else}
    \end{array}\right.
    $
   }
   & W
   \\
  \hline
  \end{tabular}}
\end{table*}

\begin{algorithm}[tbp] 
\renewcommand{\algorithmicensure}{\textbf{Initial Stage:}}
\caption{Local-MPAR Algorithm.} 
\label{alg:Local-MPAR} 
\begin{algorithmic}[1] 
\ENSURE
\STATE $n_s$ generates message $M$
\STATE $n_s.state\leftarrow$\textbf{B}
\STATE $N\leftarrow\{n_s\}$
\end{algorithmic} 
\begin{algorithmic}[1] 
\renewcommand{\algorithmicensure}{\textbf{Routing Stage:}}
\ENSURE
\FOR{any pair of $n_a$ and $n_b$}
    \IF{$n_a$ and $n_b$ encounter}
        \STATE finish the state transition according to \tablename~\ref{tab:Local-MPAR}
        \STATE update $N$
    \ENDIF
\ENDFOR
\end{algorithmic}
\begin{algorithmic}[1]
\renewcommand{\algorithmicensure}{\textbf{End Stage:}}
\ENSURE
\STATE $M$ is delivered
\end{algorithmic}
\end{algorithm}

\subsection{Tabu Search Based Movement Pattern-Aware Routing}
In Tabu-MPAR algorithm, each node also has the three states as same as in Local-MPAR. Like that in SprayAndWait routing, we allow the infectious nodes to be multiple but constrained in a finite number.

\begin{figure}[!tb]
\centering
\includegraphics[width=2in]{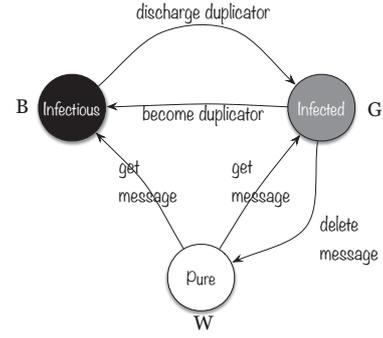}
\caption{Node state transition in Tabu-MPAR.}
\label{fig:states3}
\end{figure}

\begin{algorithm}[tbp] 
\renewcommand{\algorithmicensure}{\textbf{Initial Stage:}}
\caption{Tabu-MPAR Algorithm.} 
\label{alg:Tabu-MPAR} 
\begin{algorithmic}[1] 
\ENSURE
\STATE $n_s$ generates message $M$
\STATE $n_s.state\leftarrow$\textbf{B}
\STATE $n_s$ computes $N_{opt}$ by tabu search and saves it in $M$
\STATE $n_s.tickets\leftarrow |N_{opt}|$
\end{algorithmic} 
\begin{algorithmic}[1] 
\renewcommand{\algorithmicensure}{\textbf{Routing Stage:}}
\ENSURE
\FOR{any pair of $n_a$ and $n_b$}
    \IF{$n_a$ and $n_b$ encounter}
        \STATE $a=\frac{E[D_b]}{E[D_a]+E[D_b]}\cdot |N_{opt}|$
        \STATE $b=\frac{E[D_a]}{E[D_a]+E[D_b]}\cdot |N_{opt}|$
        \IF{$a<1$}
            \STATE $n_a.tickets=\lceil a \rceil$
            \STATE $n_b.tickets=\lfloor b \rfloor$
        \ELSE
            \STATE $n_a.tickets=\lfloor a \rfloor$
            \STATE $n_b.tickets=\lceil b \rceil$
        \ENDIF
        \STATE finish the state transition according to \tablename~\ref{tab:Tabu-MPAR}
    \ENDIF
\ENDFOR
\end{algorithmic}
\begin{algorithmic}[1]
\renewcommand{\algorithmicensure}{\textbf{End Stage:}}
\ENSURE
\STATE $M$ is delivered
\end{algorithmic}
\end{algorithm}

The Tabu-MPAR protocol is shown in algorithm~\ref{alg:Tabu-MPAR}. There are three stages as same as in Local-MPAR. In the initial stage, when the source node $n_s$ generates the message $M$, the state of $n_s$ is set to be B, as shown in line 1-2. The node set $N_{opt}$ is calculated in the source node. Besides, there is a ticket number for each generated message and is denoted by $|N_{opt}|$, which indicates that the maximal number of replicas of this message is equal to the size of $N_{opt}$. A node is in state B if and only if it has more than one tickets for this message \footnote{To say strictly, this node is currently in B state for this message, but for other message the state might be W or G}. A node with only one ticket for the message is of state G, and nodes not holding this message are in state P.

The second stage in algorithm~\ref{alg:Tabu-MPAR} shows the routing process of Tabu-MPAR. There is two main differences compared to Local-MPAR. First, there is no need to update the node set $N_{opt}$ in other node, since that it had been computed in the source node $n_s$ before and added to the message head. Second, when any two nodes $n_a$ and $n_b$ encounter, we reallocate the tickets before start the state transition process. The strategy of tickets distribution is based on $E[D_a]$ and $E[D_b]$, as shown in line 3 and 4, which ensure that the node with a smaller $E[D]$ to be allocated with more tickets. Line 5--10 guarantee that both $a$ and $b$ are an integer between $1$ and $L-1$. Then we can define the transition rule in \tablename~\ref{tab:Tabu-MPAR}. The state transition could happen in all grids except for the diagonal. Besides, we can see that the $3\times3$ table is a symmetrical matrix. Actually, when the transition in grid GB happens in $n_a$, the transition in grid BG happens in $n_b$, and vise versa. This rule also applies for WB, BW and WG, GW. The possible transitions among states are illustrated in \figurename~\ref{fig:states3}.

\begin{figure*}
\centering
\subfigure[Time $t_1$\label{local_rt1}]
{\includegraphics[width=0.24\linewidth]{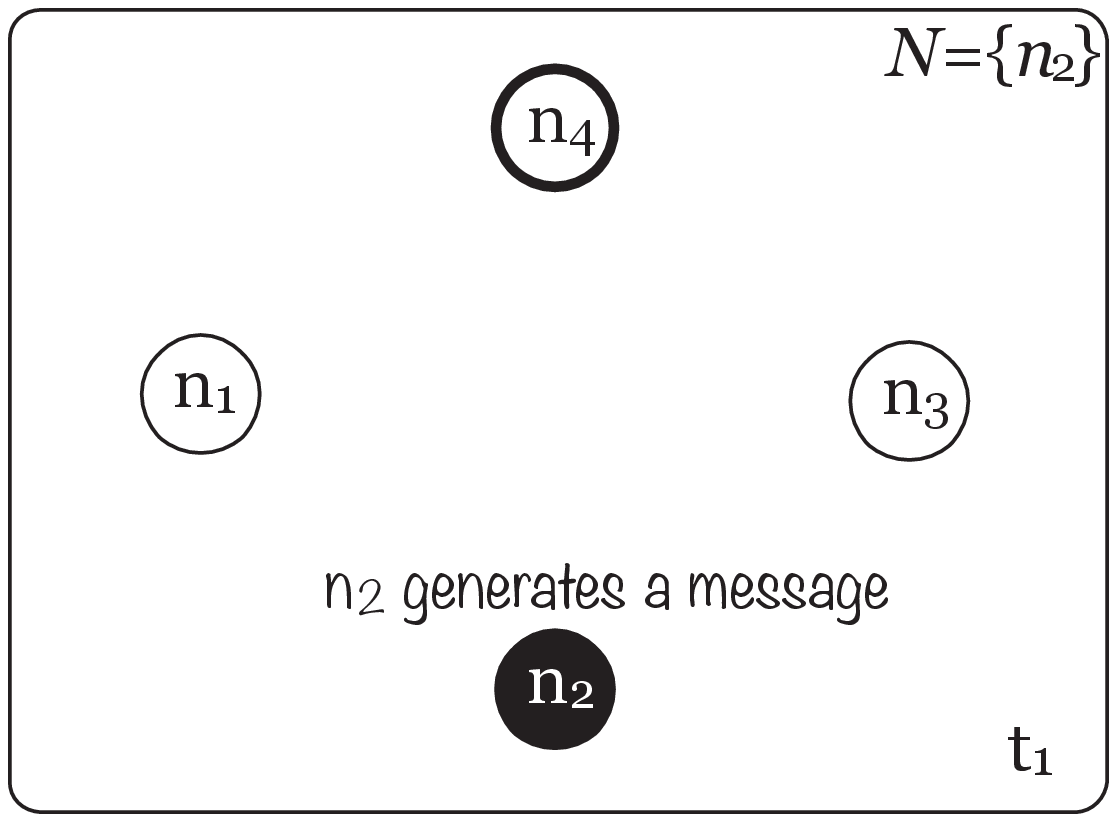}}
\subfigure[Time $t_2$\label{local_rt2}]
{\includegraphics[width=0.24\linewidth]{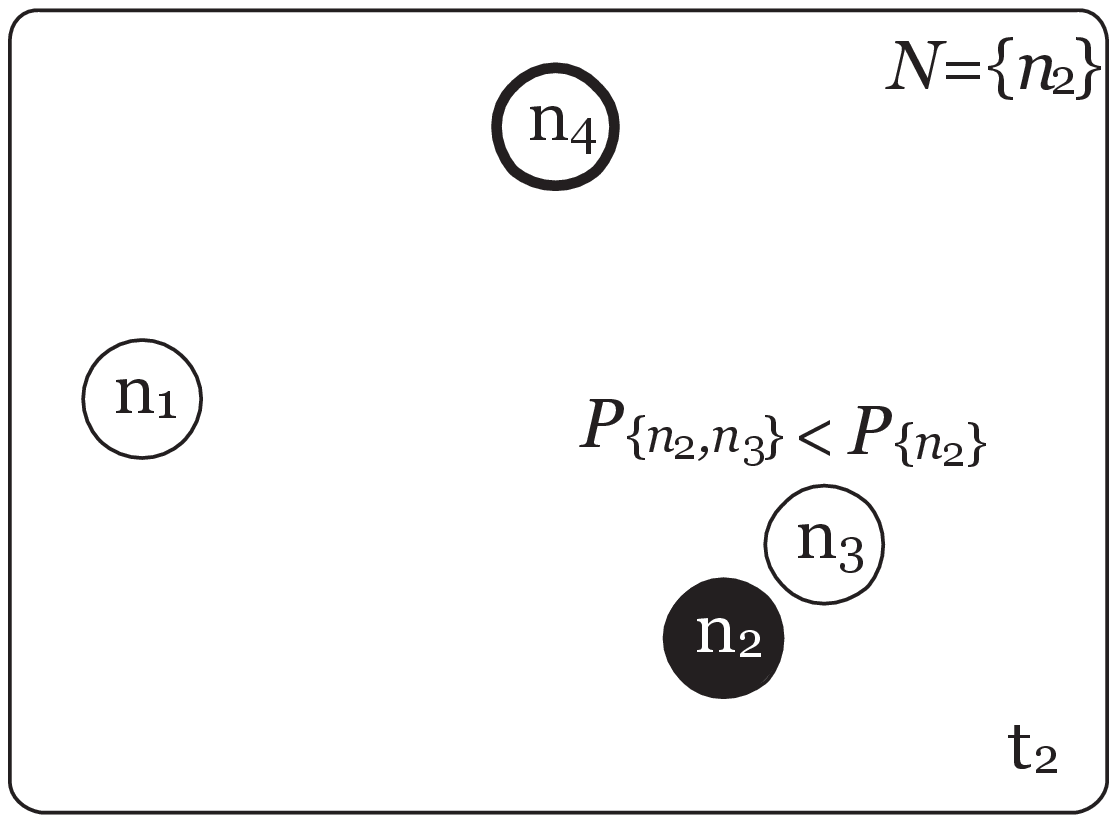}}
\subfigure[Time $t_3$\label{local_rt3}]
{\includegraphics[width=0.24\linewidth]{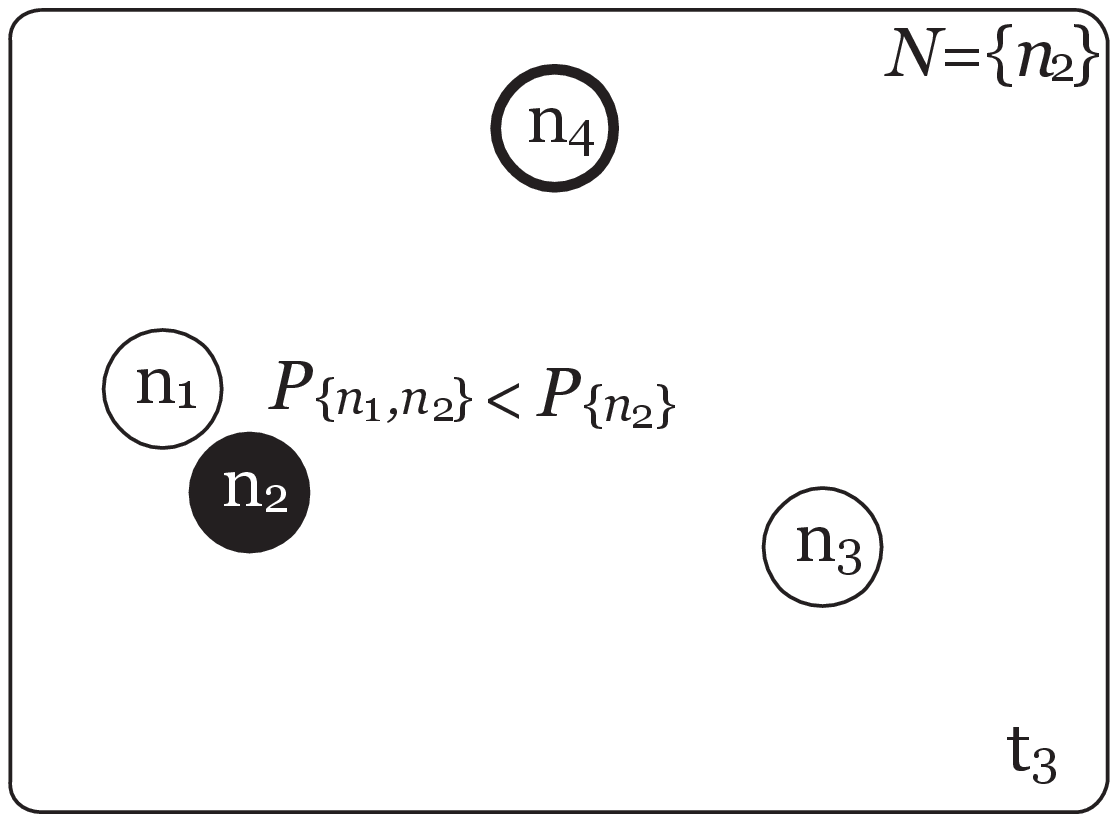}}
\subfigure[Time $t_4$\label{local_rt4}]
{\includegraphics[width=0.24\linewidth]{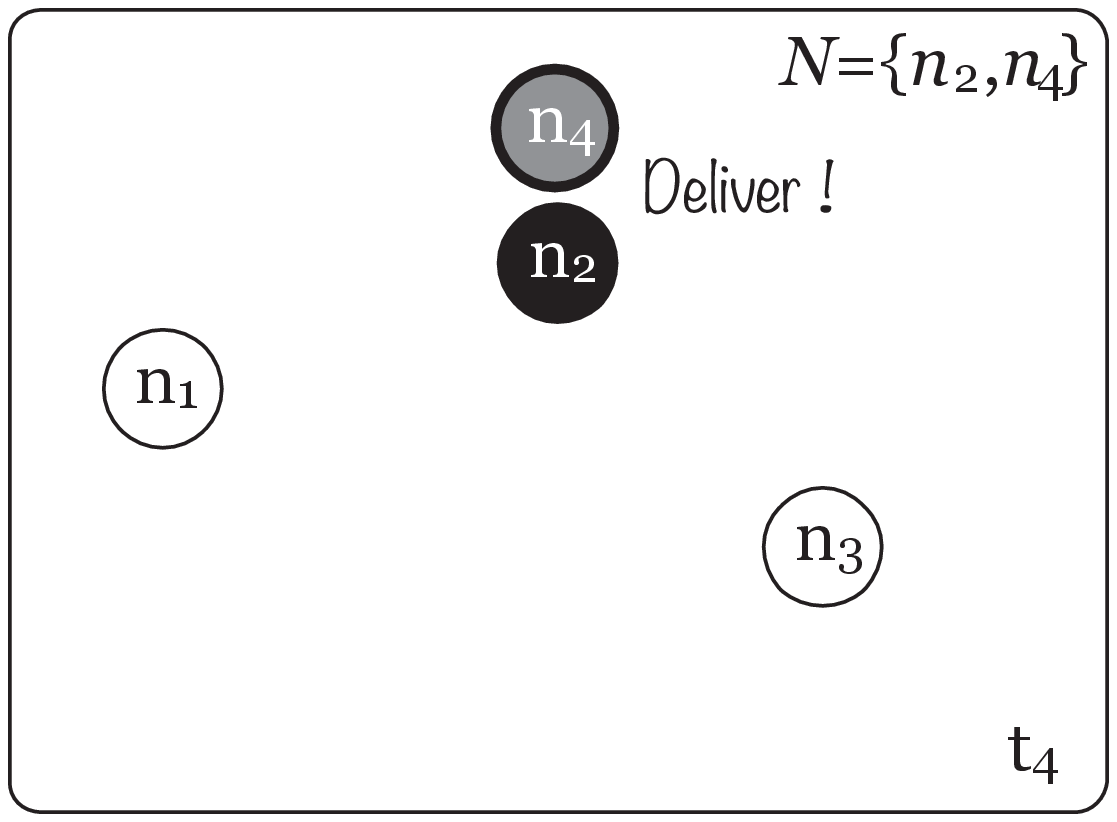}}
\caption{The routing process of Local-MPAR.}
\label{fig:local_rt}
\end{figure*}

\begin{figure*}
\centering
\subfigure[Time $t_1$\label{tabu_rt1}]
{\includegraphics[width=0.24\linewidth]{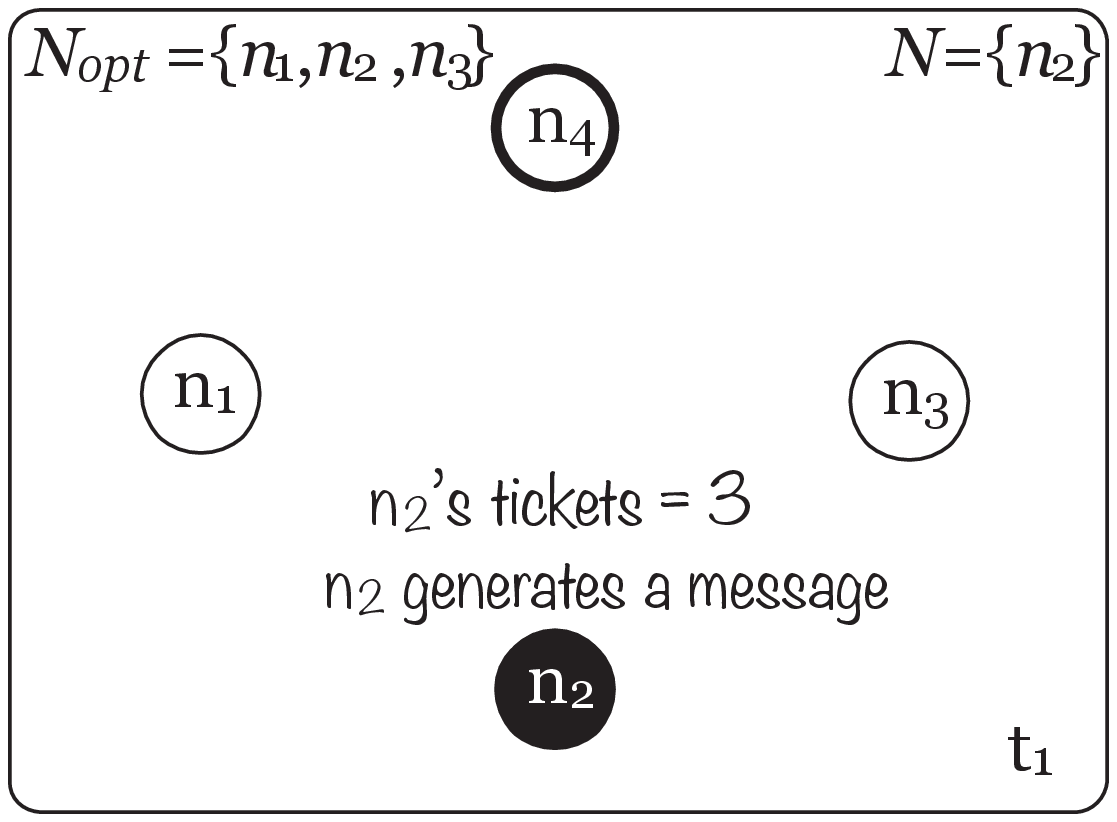}}
\subfigure[Time $t_2$\label{tabu_rt2}]
{\includegraphics[width=0.24\linewidth]{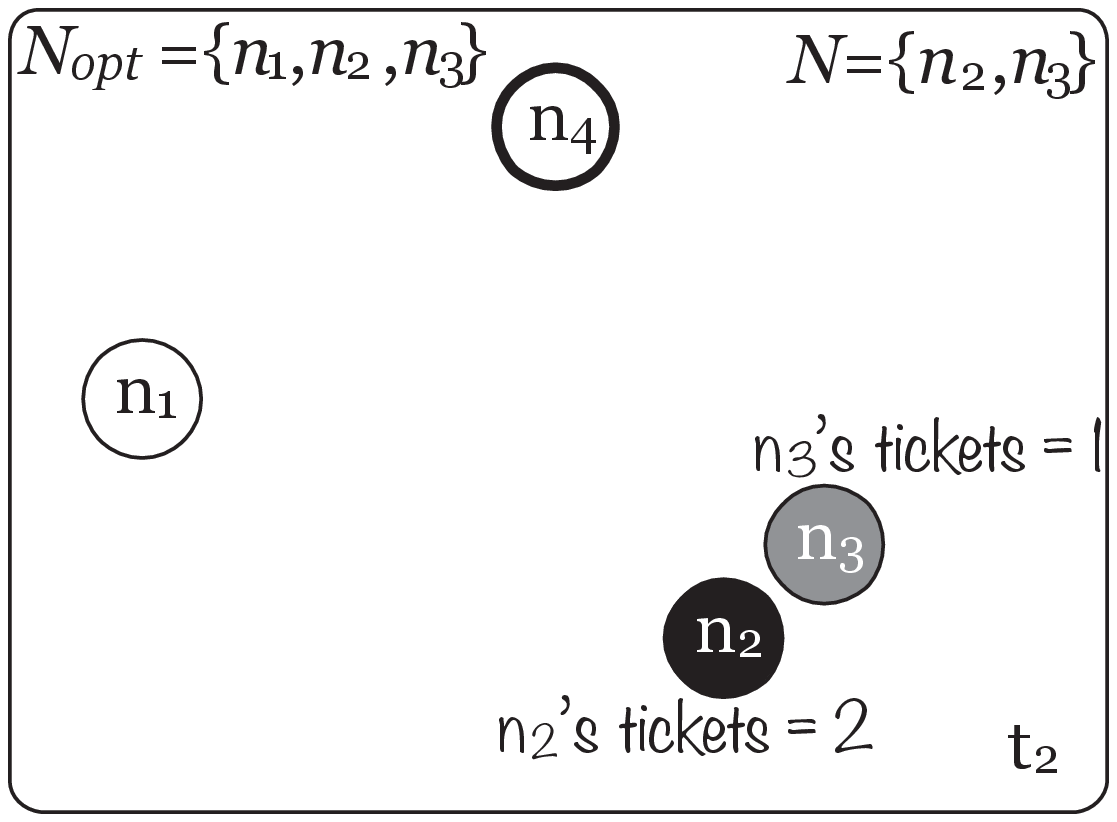}}
\subfigure[Time $t_3$\label{tabu_rt3}]
{\includegraphics[width=0.24\linewidth]{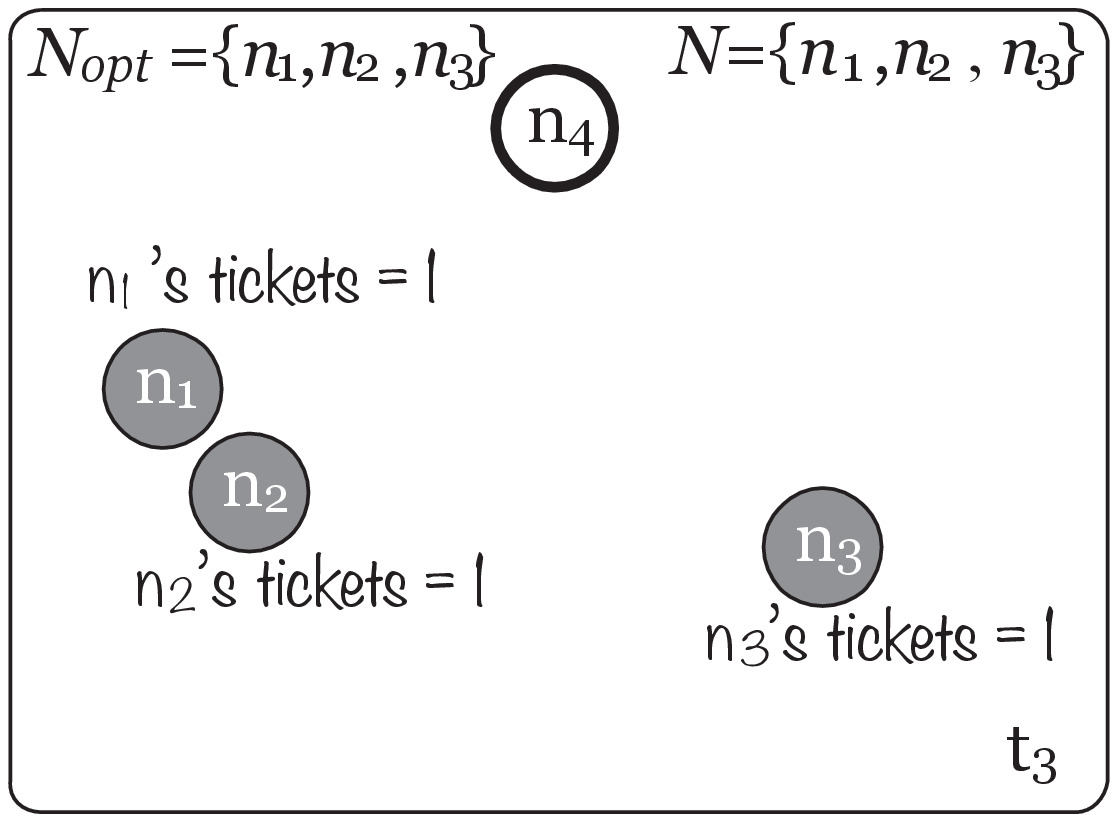}}
\subfigure[Time $t_4$\label{tabu_rt4}]
{\includegraphics[width=0.24\linewidth]{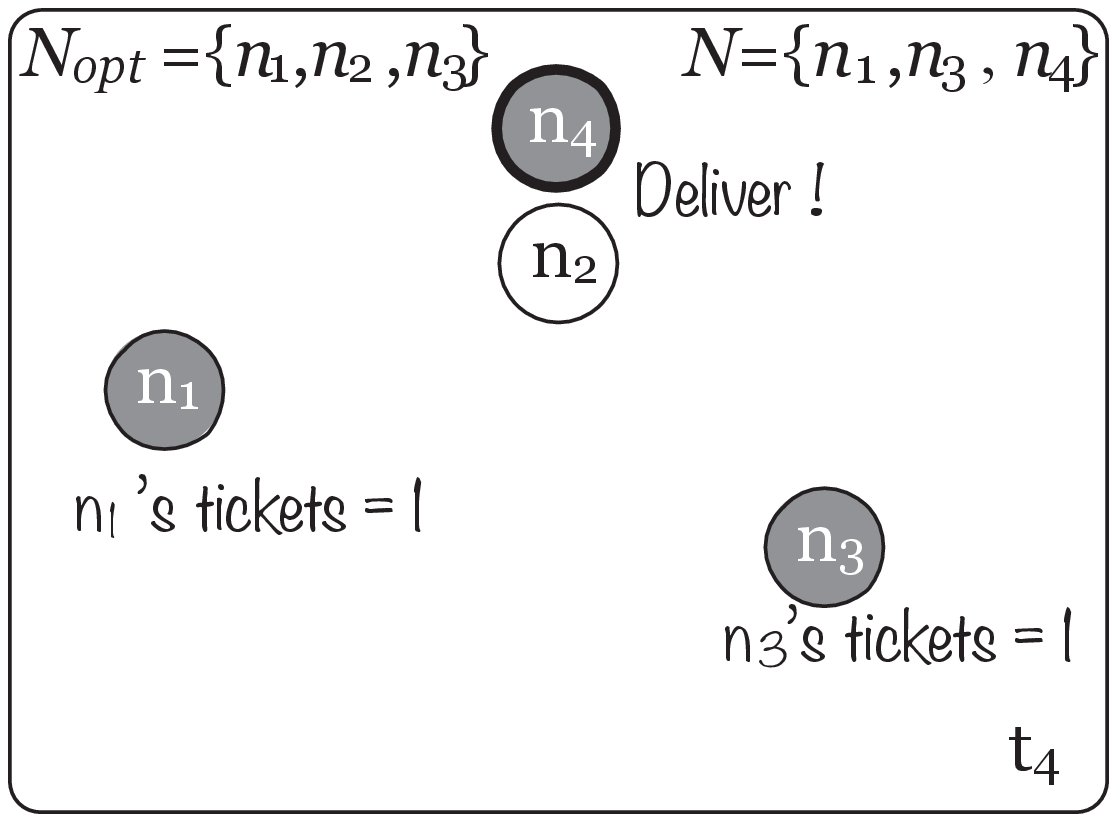}}
\caption{The routing process of Tabu-MPAR.}
\label{fig:tabu_rt}
\end{figure*}

Now we prove that the routing stage in Tabu-MPAR leads $N$ to the optimal set $N_{opt}$.

\begin{lemma}
$N=N_{opt}$ happens if and only if there is no infectious nodes in the network, i.e. all nodes in $N=N_{opt}$ are in state G.
\label{lemma:no_infectious}
\end{lemma}

\begin{proof}
Since that the number of all tickets in the network is $|N_{opt}|$, if there exist an infectious node $n_a$, then its tickets must be at least 2 and the tickets in all other nodes are less than $|N_{opt}|-2$. So the number of nodes holding the message is at most $|N_{opt}|-1$, and in this case $n_a$ has 2 tickets for the message, and all the other nodes with replicas only have 1 ticket. So the ``only if'' condition works. 
The node set holding the message accords with $N_{opt}$ means that the number of nodes holding the message is $|N_{opt}|$, which is the maximal number of replicas for this message in the network. So the number of tickets in each of these nodes is equal to 1. According to the definition of node state in Tabu-MPAR, all of these nodes are in infected state, and all other nodes are in pure state, which indicates that currently there is no infectious node in the network. So the ``if'' condition works.
\end{proof}

\begin{lemma}
If the set $N$ evolves to $N_{opt}$, then it would not change any more.
\label{lemma:change}
\end{lemma}
\begin{proof}
From lemma~2, if the set $N$ reaches $N_{opt}$, then all the nodes in $N$ are in infected (G) state, and there is no non-pure node outside $N$. Since that there is no node in B state, so the transition in grid GB could not happen any more. In grid GW, the transition $G\rightarrow W$ happens only if the current node is not belong to $N_{opt}=N$, which conflicts the assumption. So all the states of all elements in $N$ no longer change.
\end{proof}

\begin{theorem}
The transition rule in \tablename~\ref{tab:Tabu-MPAR} leads $N$ to $N_{opt}$.
\end{theorem}
\begin{proof}
The state transition between B and G does not modify the set $N$. However, from line 4--10 in algorithm~3, it is guaranteed that the distributed tickets are no less than 1, which means that any node in state B would transfer to state G if there are enough encounter opportunities during the routing process. So we just need to focus on the transitions between W and G, as shown in grid WG and GW, where the rule is to transfer the node belonging to $N_{opt}$ to G. Consequently our theorem comes directly from lemma~2 and 3.
\end{proof}

\figurename~\ref{fig:tabu_rt} shows the routing process of Tabu-MPAR, the network situation is the same as \figurename~\ref{fig:local_rt}, where the source $n_2$ generates the message at time $t_1$. Since $n_2$ computes that $N_{opt}=\{n_1,n_2,n_3\}$, so the number of tickets for $n_2$ is initialized to be $|N_{opt}|=3$. At time $t_2$, when $n_2$ and $n_3$ encountered, we allocate the tickets by referring to $E[D_2]$ and $E[D_3]$. The node set $N$ has evolved to $N_{opt}$ at time $t_3$, and there is no infectious nodes in the network, corresponding to lemma~\ref{lemma:no_infectious}. Finally the routing process ends in $t_4$ due to the completion of the delivery.

\section{Simulation}
\label{sec:simulation}
In this section, we conduct extensive simulations to evaluate the performance of the MPAR algorithm by using the ONE simulator \cite{Keranen2009}. We employ the Working Day Movement (WDM) model proposed in \cite{Ekman2008} and set the map to be Manhattan blocks. This model incorporates some sense of hierarchy and distinguishes between inter-building and intra-building movements. Some sub-models such as home, office, evening activities and different transports are introduced in WDM, so as to capture the society characteristics of people. An office model reproduces a kind of star-like trajectories around a desk of the person inside the office building, while home model is just a sojourn in a particular point of a home location. The evening activity sub-model reflects a meeting with friends after work by modeling a random walk of a group along the streets. The Manhattan community model confines the routing paths of the mobile nodes to certain paths that reflect their real moving pattern in addition to the colocation pattern. The Manhattan model consists of grids in a matrix, in which all nodes can only move on the sides of a grid, as shown in \figurename~\ref{fig:manhattan}.

\begin{figure}[!t]
\centering
\includegraphics[width=2in]{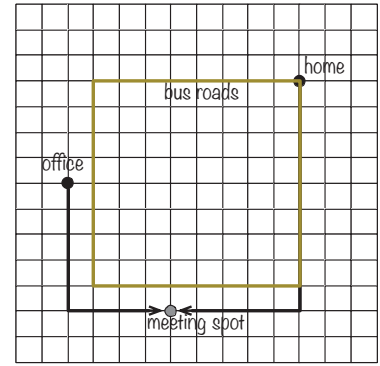}
\caption{Working day movement model in Manhattan blocks}
\label{fig:manhattan}
\end{figure}

Concretely, the network is established as follows. We set the number of pedestrians to be 200, 400, 600, 800. The number of buses, offices and activity spots are correspondingly set to be 2\%, 20\%, 2.5\% of the total number of pedestrians. The working time for each simulation day constants at 8 simulation hours for each person. Each person's probability to do evening activity is set to be 0.5. Besides, we set the probability of owning a car for each person to be 0.2, so people might drive themselves to these places instead of taking public vehicles. The walking speed of pedestrians is set to be 0.8--1.4 m/s. As shown in \figurename~\ref{fig:manhattan}, there are three kinds of locations in the simulation, i.e. home, office and meeting spot. There is a loop line for buses in the center of the entire Manhattan square, so the pedestrians can take a bus to reach their destinations. We put a throw box in every home, office and meeting spot. The capacity of each throwbox is set to be sufficient enough to take custody of all received messages. All the nodes communicate with the bluetooth interface, of which the transmit diameter is 10m and the transmit speed is 250KB/s. Since that each spot (home, office and meeting spot) is set to be a point in the map, if a node reach the spot, it is ensured to encounter with the throwbox so as to have its messages taken custody. 

 We compare the MPAR algorithm with the existing social-aware algorithms: Delegation Forwarding (DF) \cite{Erramilli2008} and SimBet routing \cite{Daly2007}. 
Four performance metrics commonly used are examined: delivery ratio, average latency overhead ratio and average hop count. Simulation results demonstrate that the MPAR algorithm can significantly improve SDTNs routing performance.

\subsection{Simulation in Working Day Movement Model}
In the simulations, 30,000 messages are generated by randomly selecting source and destination nodes among all the pedestrians. The size of each message is between 500K and 1M. In simulations on evaluating the four metrics, we set $|\overline{N}|$=200, 400, 600 and 800, respectively. The simulation area is adjusted according to $|\overline{N}|$ to keep the connectivity of the network. The entire simulation time is set to be 12 days. 
\subsubsection{Varying Time-to-live}

\begin{figure*}[!t]
\centering
\subfigure[Number of nodes: $|\overline{N}|=200$\label{200_delivery_ttl}]
{\includegraphics[width=0.24\textwidth]{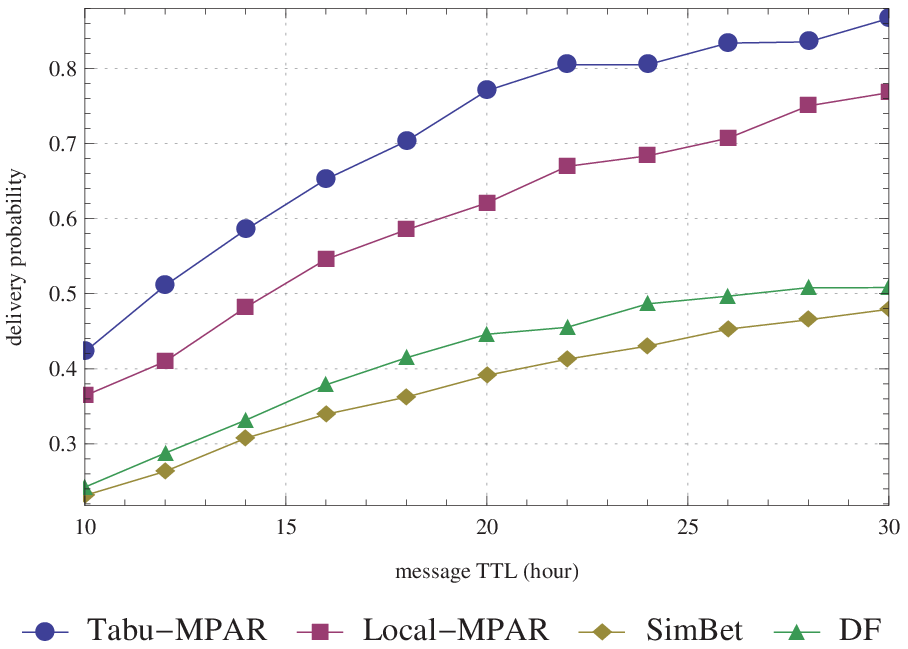}}
\subfigure[Number of nodes: $|\overline{N}|=400$\label{400_delivery_ttl}]
{\includegraphics[width=0.24\textwidth]{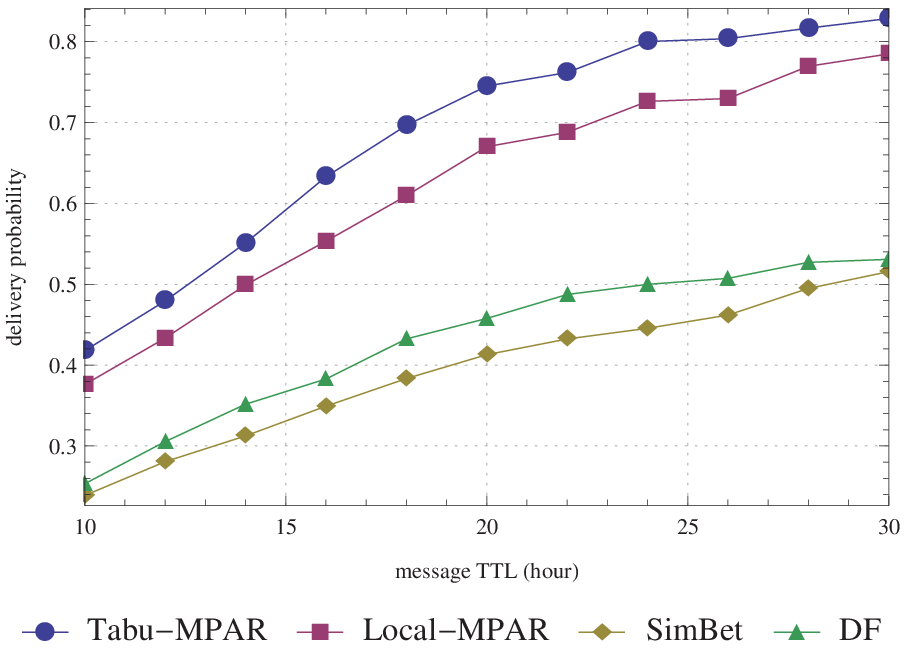}}
\subfigure[Number of nodes: $|\overline{N}|=600$\label{600_delivery_ttl}]
{\includegraphics[width=0.24\textwidth]{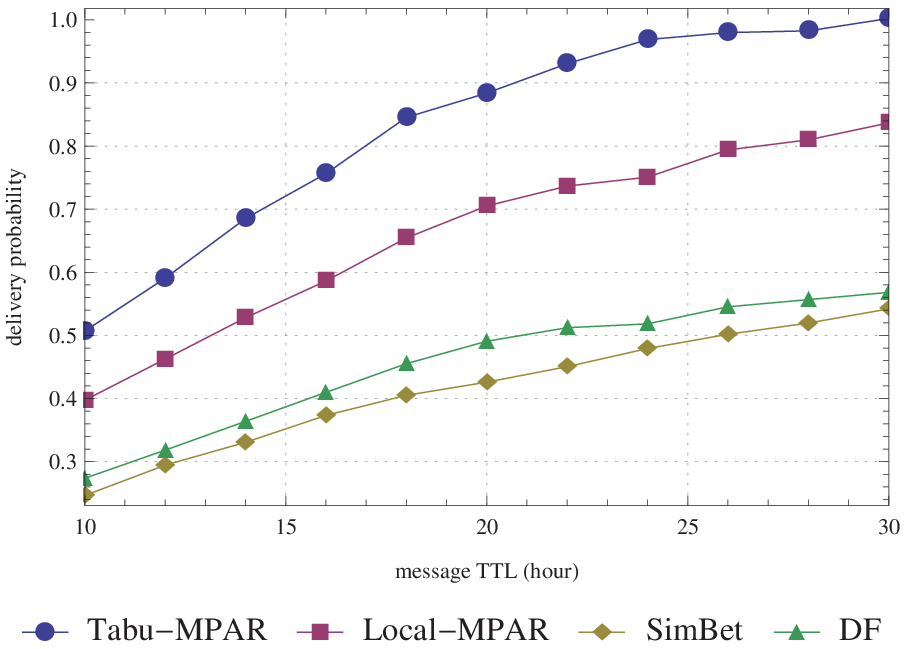}}
\subfigure[Number of nodes: $|\overline{N}|=800$\label{800_delivery_ttl}]
{\includegraphics[width=0.24\textwidth]{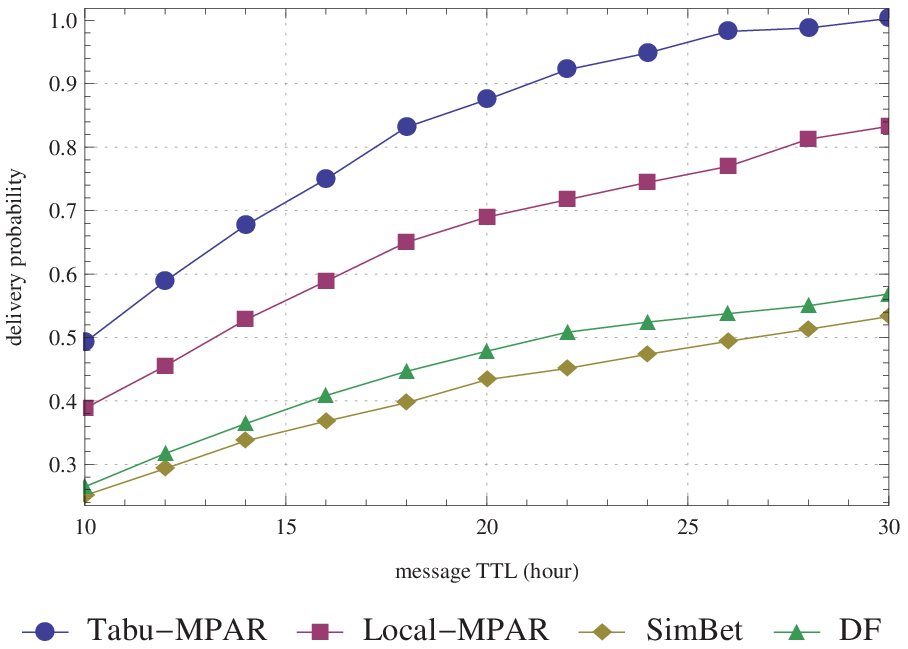}}
\caption{Performance comparisons of delivery probability vs. time-to-live}
\label{fig:delivery_ttl}
\end{figure*}

\begin{figure*}[!t]
\centering
\subfigure[Number of nodes: $|\overline{N}|=200$\label{200_avgLatency_ttl}]
{\includegraphics[width=0.24\textwidth]{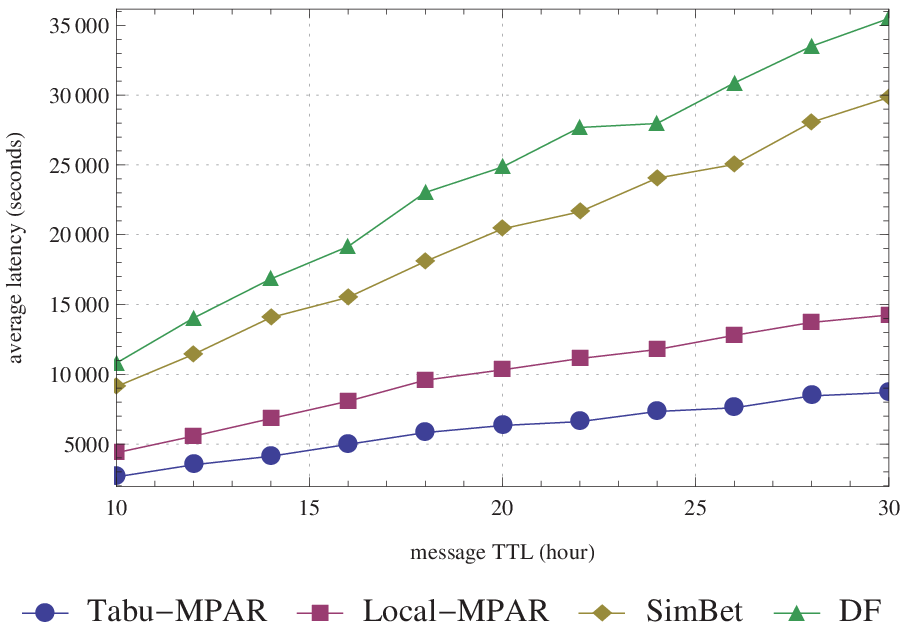}}
\subfigure[Number of nodes: $|\overline{N}|=400$\label{400_avgLatency_ttl}]
{\includegraphics[width=0.24\textwidth]{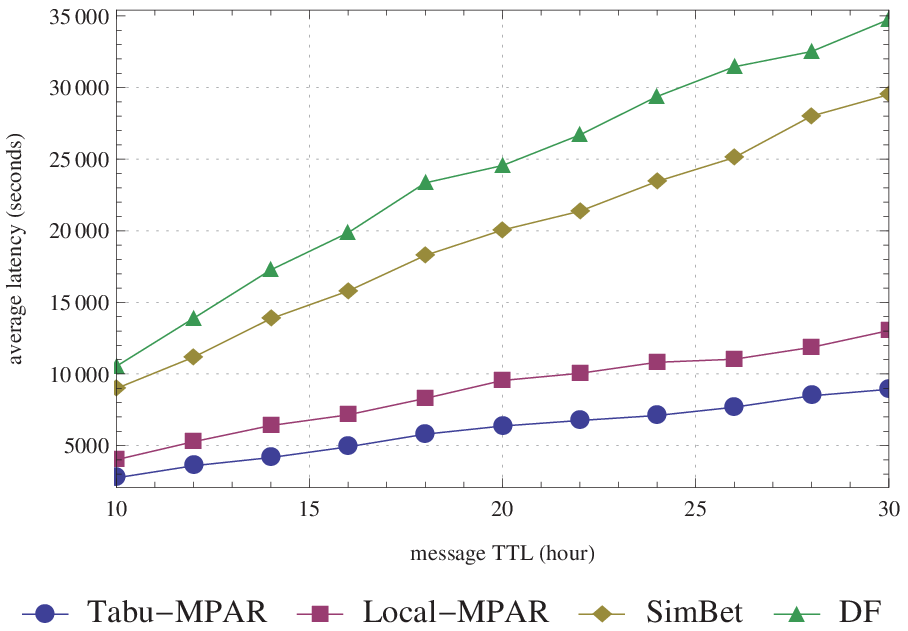}}
\subfigure[Number of nodes: $|\overline{N}|=600$\label{600_avgLatency_ttl}]
{\includegraphics[width=0.24\textwidth]{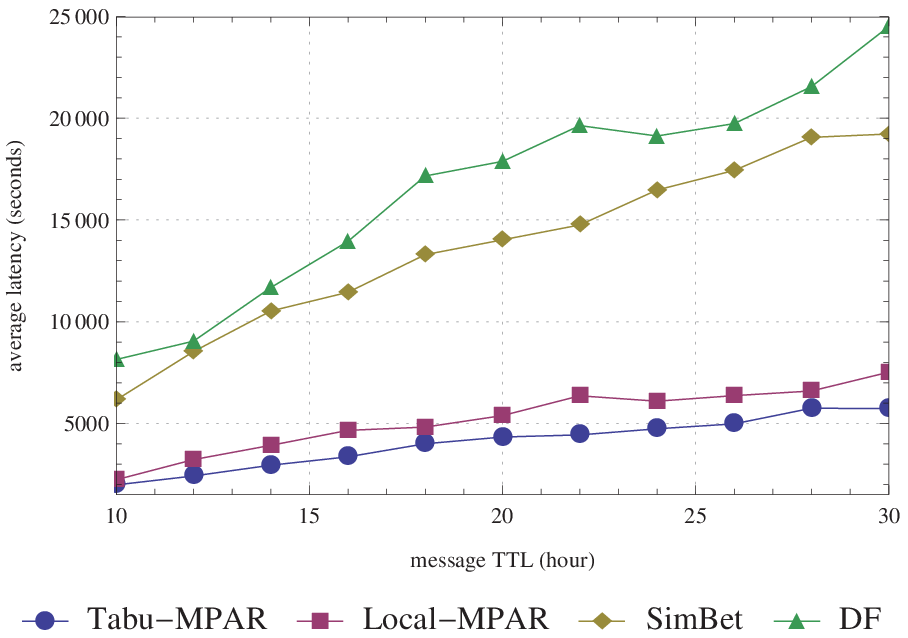}}
\subfigure[Number of nodes: $|\overline{N}|=800$\label{800_avgLatency_ttl}]
{\includegraphics[width=0.24\textwidth]{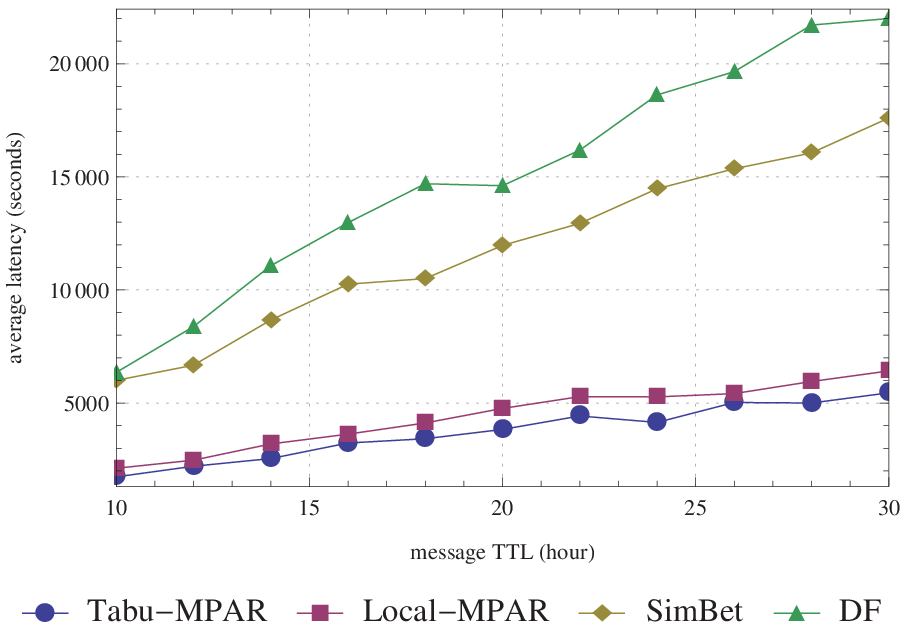}}
\caption{Performance comparisons of average end-to-end latency vs. time-to-live}
\label{fig:latency_ttl}
\end{figure*}

\begin{figure*}[!t]
\centering
\subfigure[Number of nodes: $|\overline{N}|=200$\label{200_overhead_ttl}]
{\includegraphics[width=0.24\textwidth]{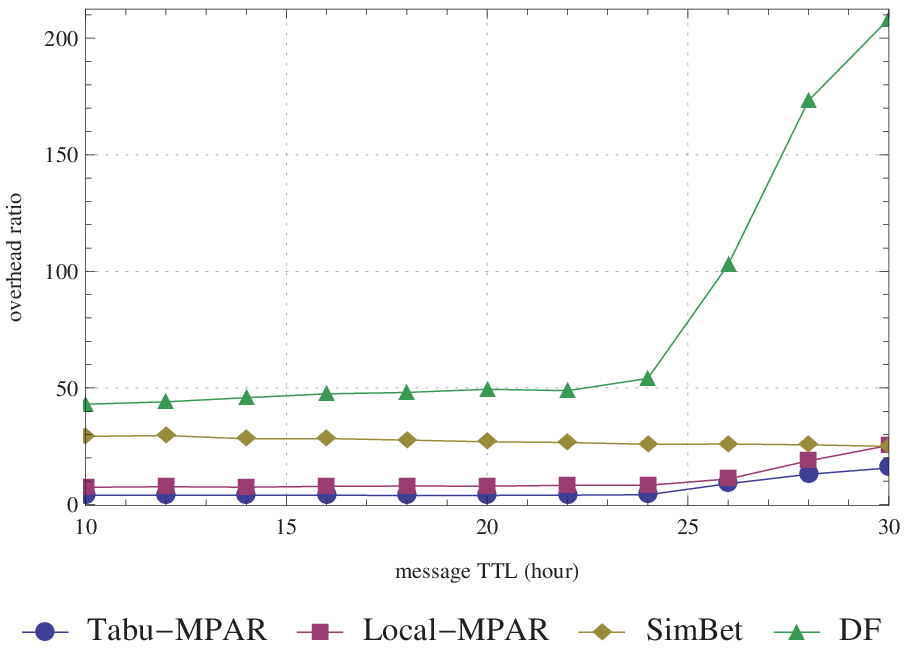}}
\subfigure[Number of nodes: $|\overline{N}|=400$\label{400_overhead_ttl}]
{\includegraphics[width=0.24\textwidth]{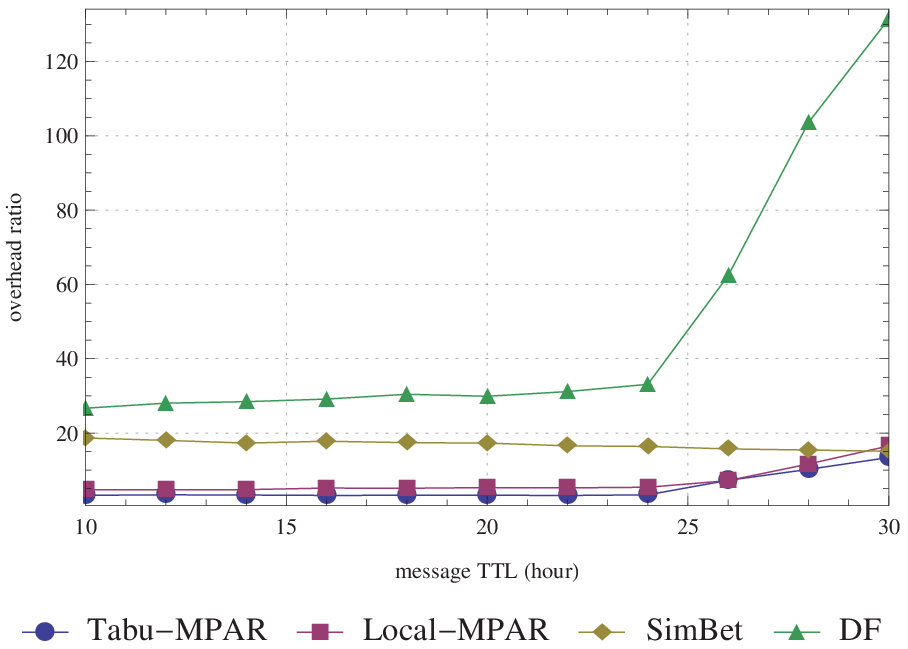}}
\subfigure[Number of nodes: $|\overline{N}|=600$\label{600_overhead_ttl}]
{\includegraphics[width=0.24\textwidth]{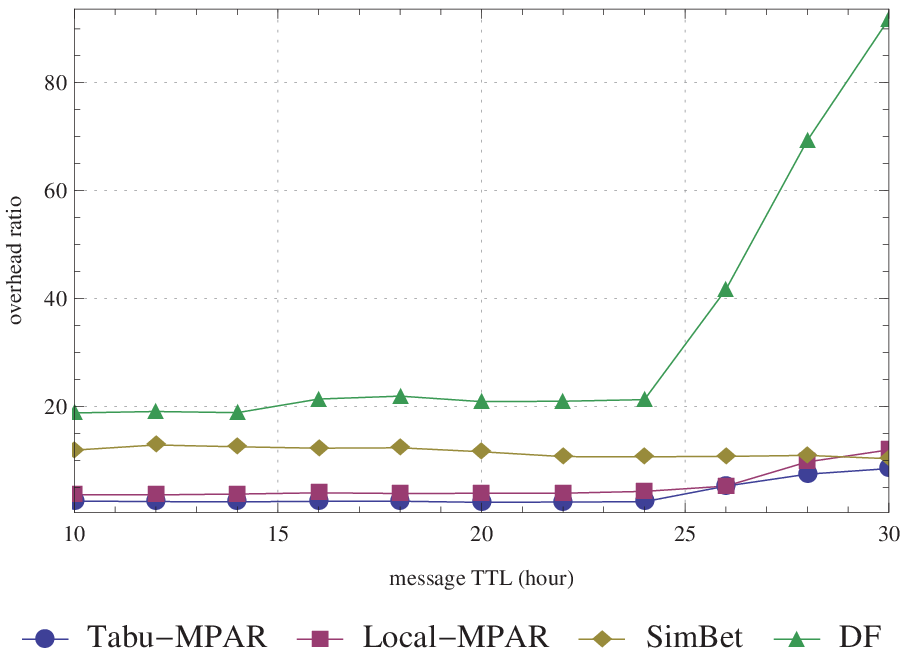}}
\subfigure[Number of nodes: $|\overline{N}|=800$\label{800_overhead_ttl}]
{\includegraphics[width=0.24\textwidth]{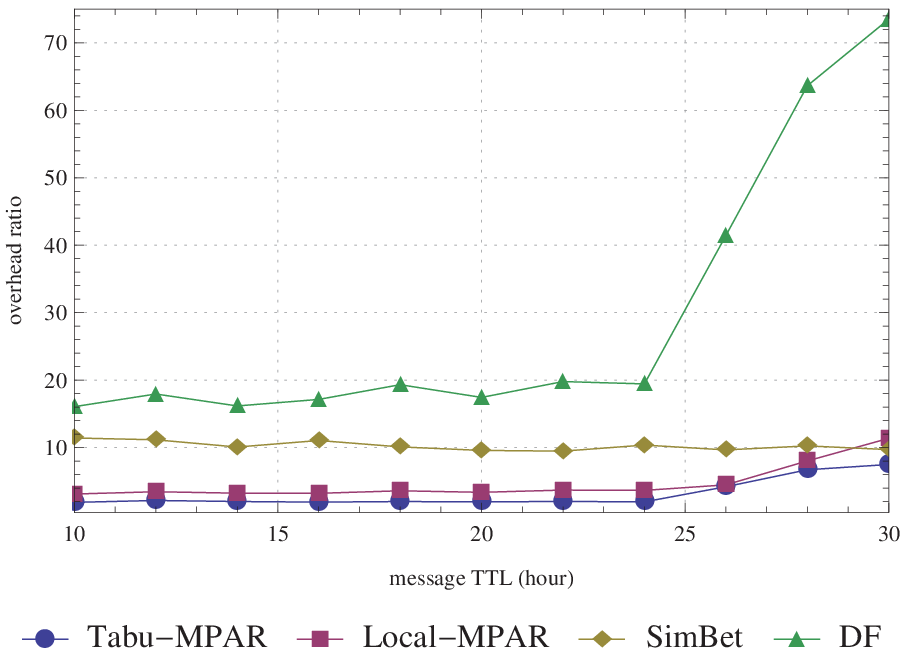}}
\caption{Performance comparisons of overhead ratio vs. time-to-live}
\label{fig:overhead_ttl}
\end{figure*}

\begin{figure*}[!t]
\centering
\subfigure[Number of nodes: $|\overline{N}|=200$\label{200_avgHop_ttl}]
{\includegraphics[width=0.24\textwidth]{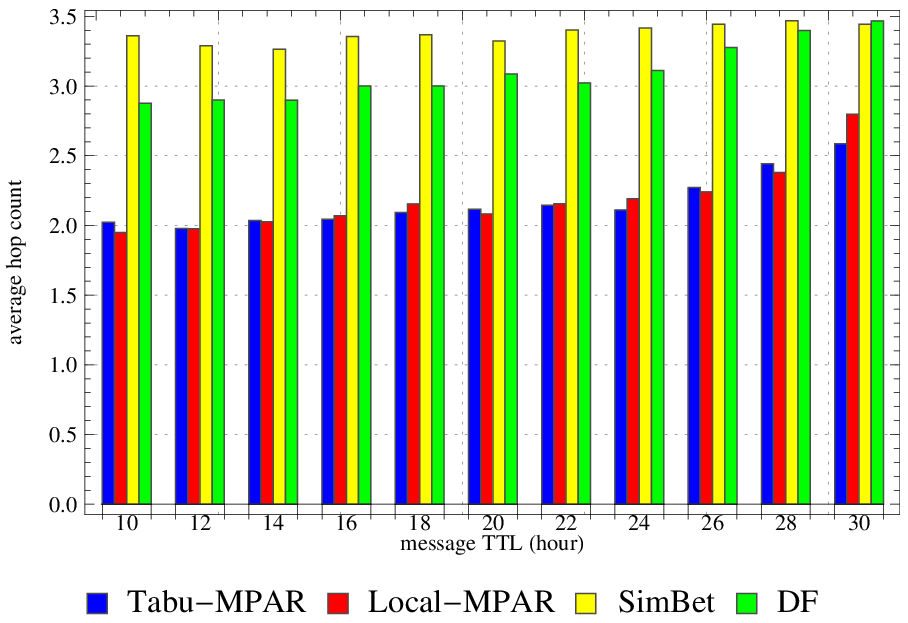}}
\subfigure[Number of nodes: $|\overline{N}|=400$\label{400_avgHop_ttl}]
{\includegraphics[width=0.24\textwidth]{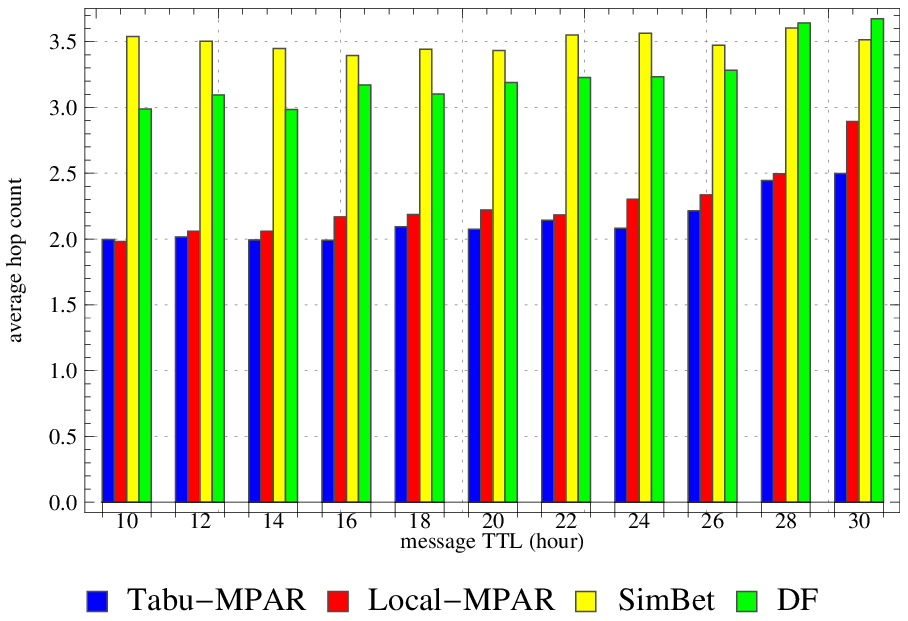}}
\subfigure[Number of nodes: $|\overline{N}|=600$\label{600_avgHop_ttl}]
{\includegraphics[width=0.24\textwidth]{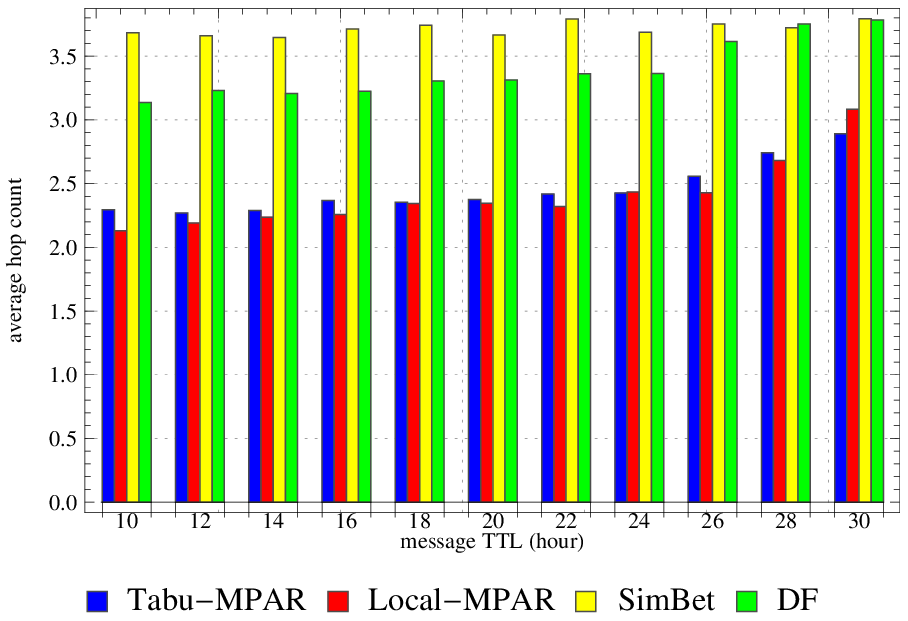}}
\subfigure[Number of nodes: $|\overline{N}|=800$\label{800_avgHop_ttl}]
{\includegraphics[width=0.24\textwidth]{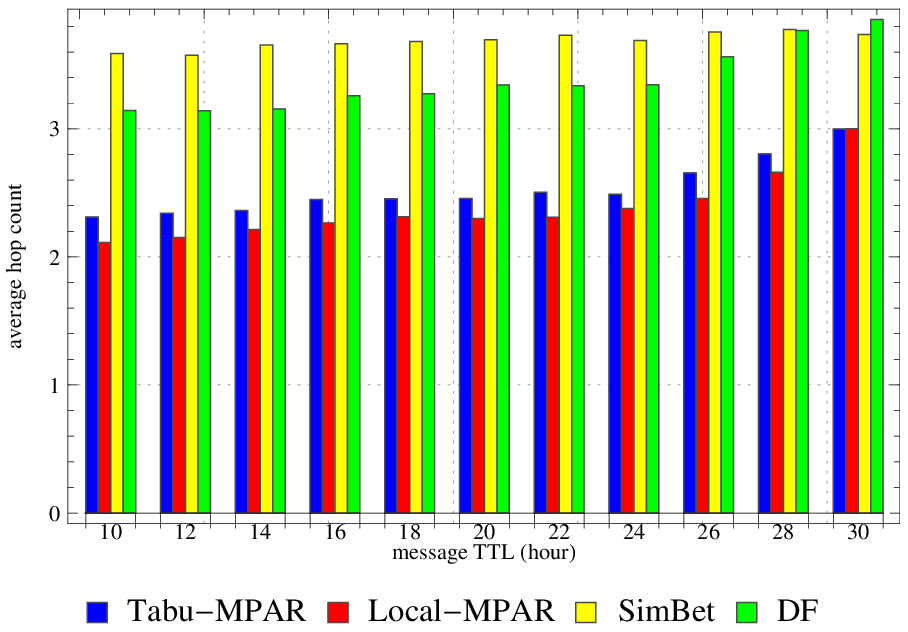}}
\caption{Performance comparisons of average hop count vs. time-to-live}
\label{fig:hop_ttl}
\end{figure*}

In the simulation of varying the time-to-live value, we set the buffer size constant at 200M. The time-to-live value is set from 10 to 30 hours. The simulation results are shown in \figurename~\ref{fig:delivery_ttl}, \ref{fig:latency_ttl}, \ref{fig:overhead_ttl} and \ref{fig:hop_ttl}.

These results show that the two MPAR algorithms significantly outperform DF and SimBet. Compared with DF, Tabu-MPAR and Local-MPAR increase the delivery ratio by about 71.1\% and 37.8\%, and reduces the average latency by approximate 79.2\% and 60.1\%, respectively. Compared with SimBet, Tabu-MPAR and Local-MPAR increase the delivery ratio by about 95.2\% and 55.3\%, and reduces the average latency by about 50\% and 70\%. Regarding the results in \figurename~\ref{fig:overhead_ttl}, both Tabu-MPAR and Local-MPAR have a comparably lower overhead ratio than DF and SimBet. For the three social-aware routing algorithm Tabu-MPAR, Local-MPAR and SimBet, the overhead performance is much better than DF, especially when the message TTL is set to be relatively large. When the number of nodes is smaller, the improvement is a little more apparent. Since that the MPAR routing algorithms constrain the maximal number of replicas and utilize the throwbox to deliver messages, the number of relay operations is very low. Moreover, along with the increasing of nodes, the overloads for each node is released, so the overhead ratio becomes lower. \figurename~\ref{fig:hop_ttl} shows the average hop count performance. We can see that the two MPAR algorithms have smaller average hop count than the other two ones. With the number of nodes increasing, the average hop count increases. This is because the simulation area is correspondingly magnified and it needs more nodes to cooperate to delivery each message.

\subsubsection{Varying Buffer Size}
In the simulation of varying the buffer size, we set the message time-to-live constant at 16 simulation hours. The buffer size is set from 50 to 300 MB. The simulation results are shown in \figurename~\ref{fig:delivery_buffer}, \ref{fig:latency_buffer}, \ref{fig:overhead_buffer} and \ref{fig:hop_buffer}.

\begin{figure*}[!t]
\centering
\subfigure[Number of nodes: $|\overline{N}|=200$\label{200_delivery_buffer}]
{\includegraphics[width=0.24\textwidth]{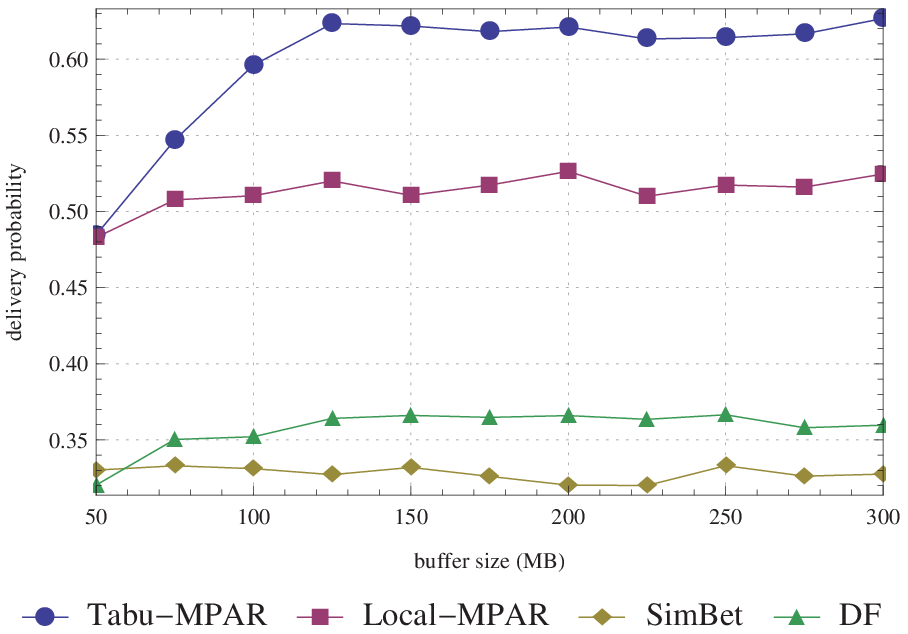}}
\subfigure[Number of nodes: $|\overline{N}|=400$\label{400_delivery_buffer}]
{\includegraphics[width=0.24\textwidth]{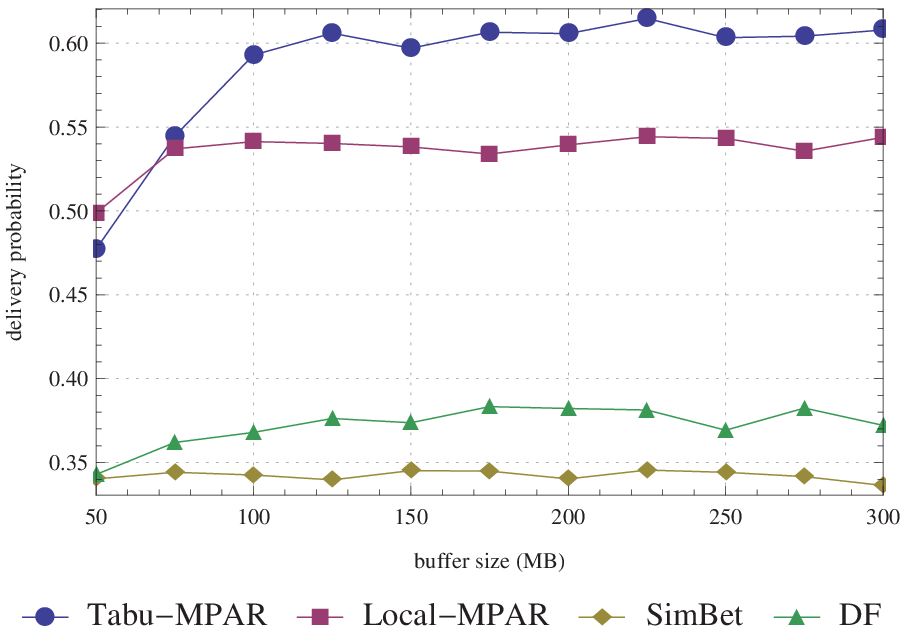}}
\subfigure[Number of nodes: $|\overline{N}|=600$\label{600_delivery_buffer}]
{\includegraphics[width=0.24\textwidth]{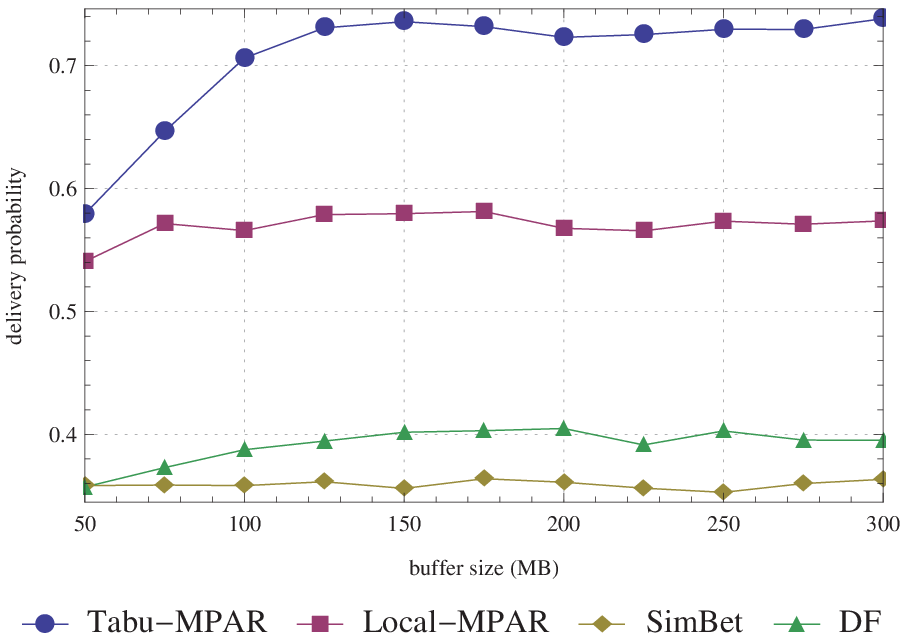}}
\subfigure[Number of nodes: $|\overline{N}|=800$\label{800_delivery_buffer}]
{\includegraphics[width=0.24\textwidth]{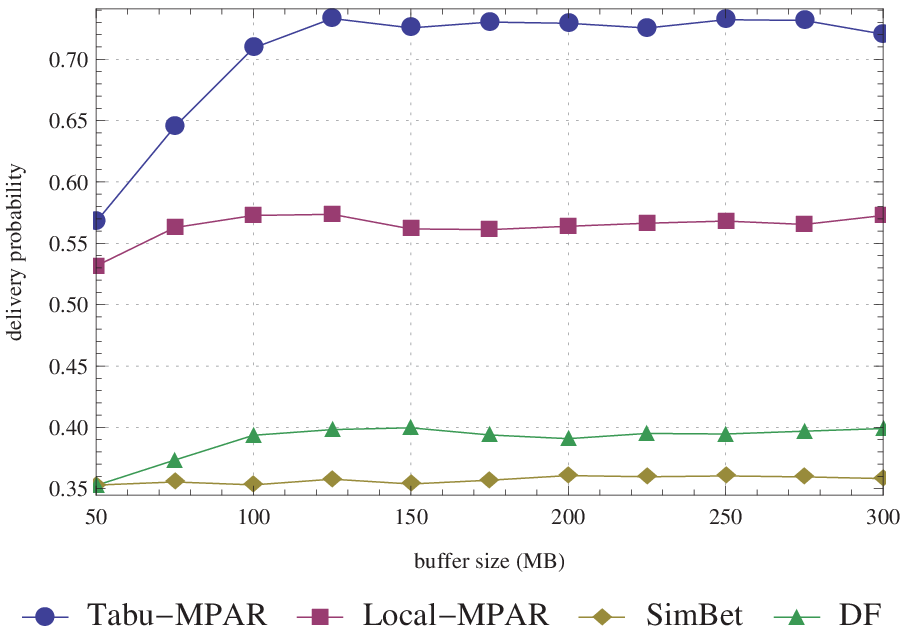}}
\caption{Performance comparisons of delivery probability vs. buffer size}
\label{fig:delivery_buffer}
\end{figure*}

\begin{figure*}[!t]
\centering
\subfigure[Number of nodes: $|\overline{N}|=200$\label{200_avgLatency_buffer}]
{\includegraphics[width=0.24\textwidth]{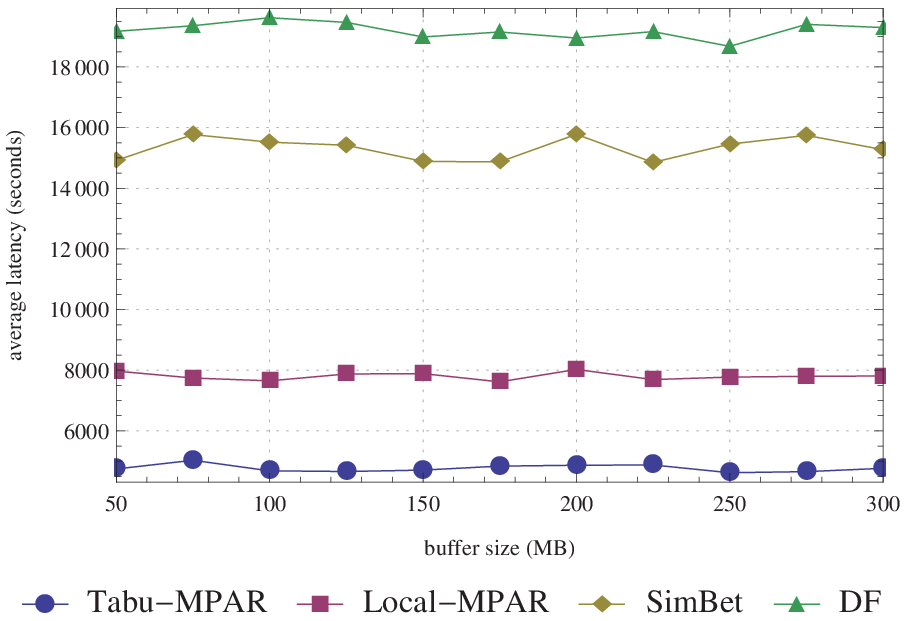}}
\subfigure[Number of nodes: $|\overline{N}|=400$\label{400_avgLatency_buffer}]
{\includegraphics[width=0.24\textwidth]{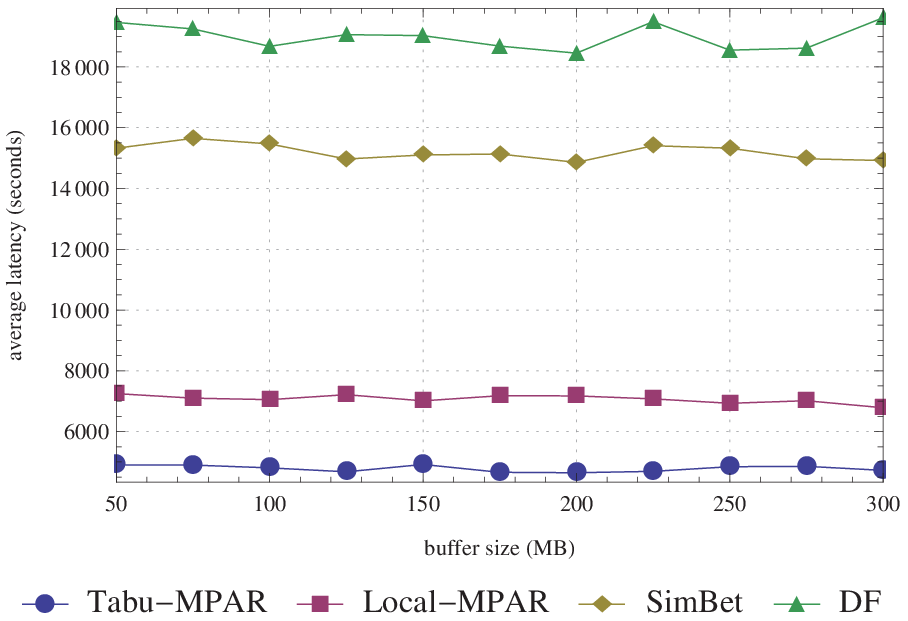}}
\subfigure[Number of nodes: $|\overline{N}|=600$\label{600_avgLatency_buffer}]
{\includegraphics[width=0.24\textwidth]{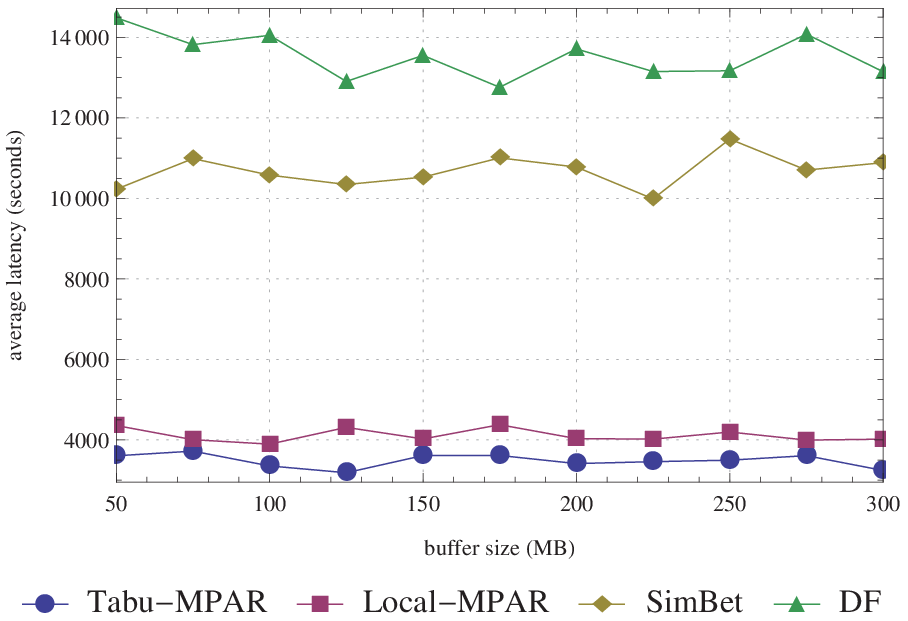}}
\subfigure[Number of nodes: $|\overline{N}|=800$\label{800_avgLatency_buffer}]
{\includegraphics[width=0.24\textwidth]{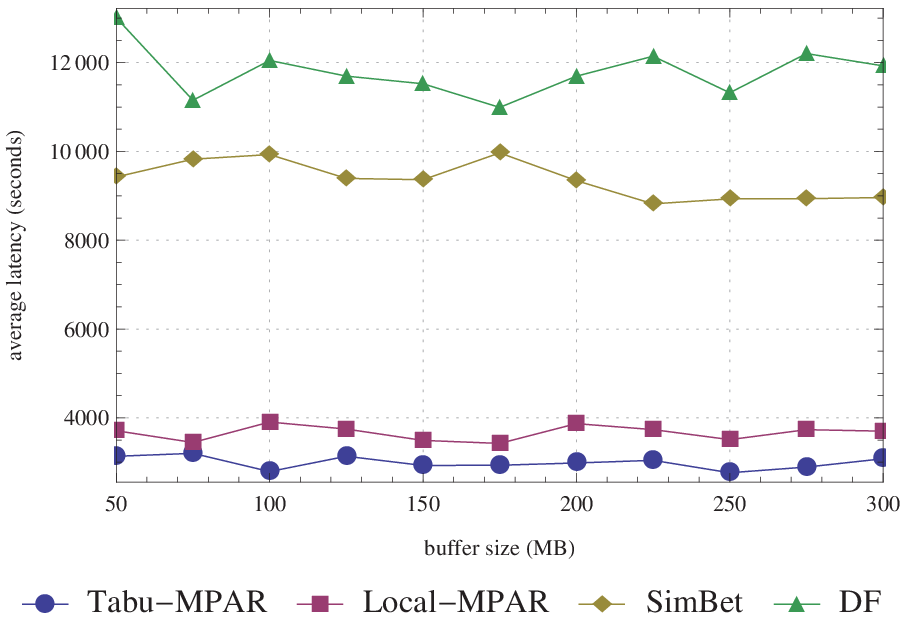}}
\caption{Performance comparisons of delivery probability vs. buffer size}
\label{fig:latency_buffer}
\end{figure*}

\begin{figure*}[!t]
\centering
\subfigure[Number of nodes: $|\overline{N}|=200$\label{200_overhead_buffer}]
{\includegraphics[width=0.24\textwidth]{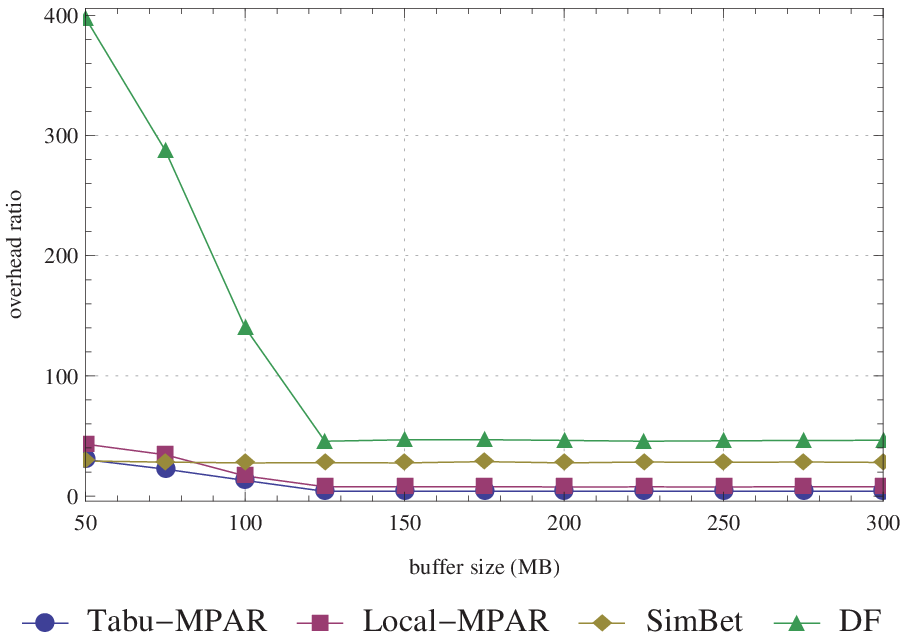}}
\subfigure[Number of nodes: $|\overline{N}|=400$\label{400_overhead_buffer}]
{\includegraphics[width=0.24\textwidth]{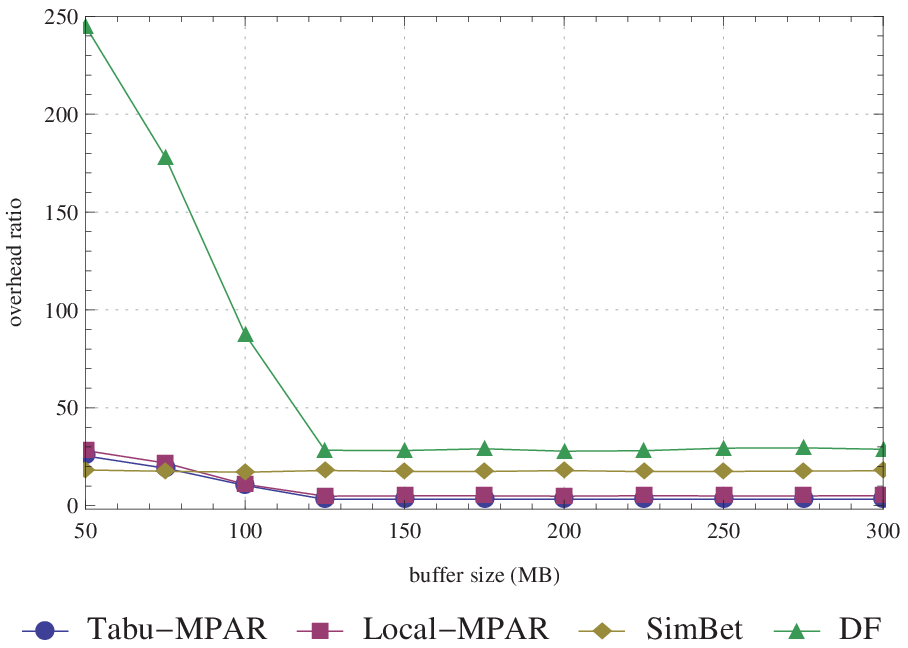}}
\subfigure[Number of nodes: $|\overline{N}|=600$\label{600_overhead_buffer}]
{\includegraphics[width=0.24\textwidth]{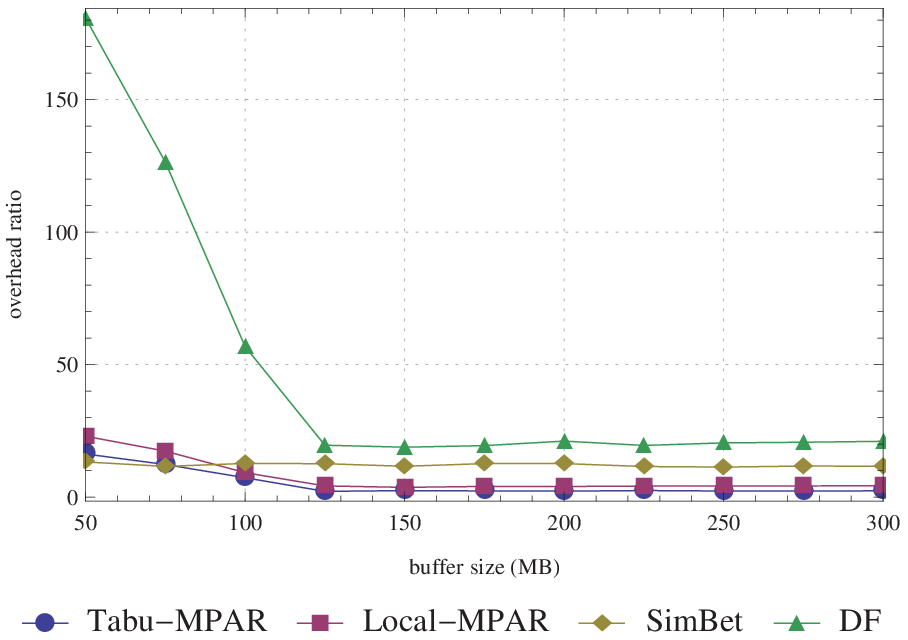}}
\subfigure[Number of nodes: $|\overline{N}|=800$\label{800_overhead_buffer}]
{\includegraphics[width=0.24\textwidth]{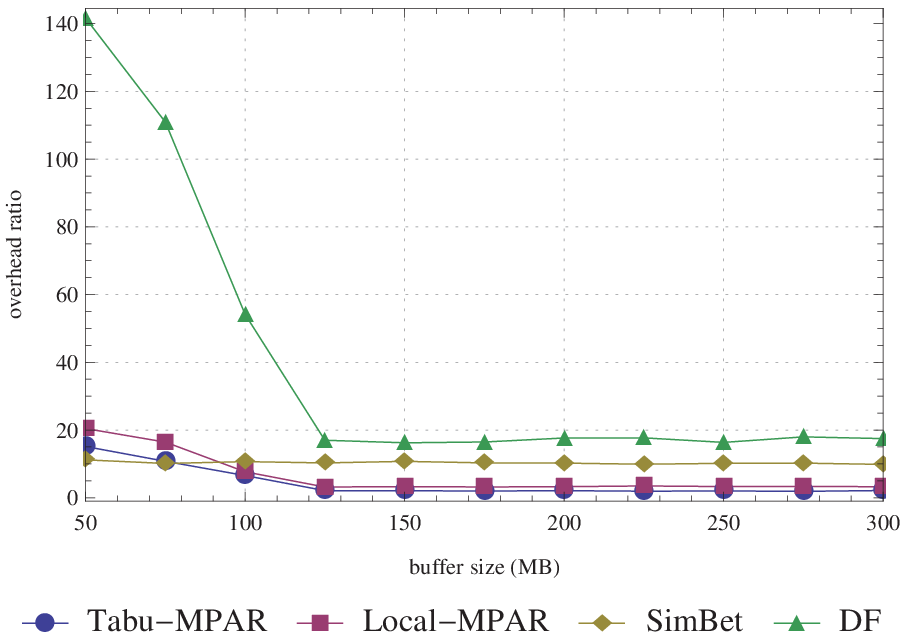}}
\caption{Performance comparisons of delivery probability vs. buffer size}
\label{fig:overhead_buffer}
\end{figure*}

\begin{figure*}[!t]
\centering
\subfigure[Number of nodes: $|\overline{N}|=200$\label{200_avgHop_buffer}]
{\includegraphics[width=0.24\textwidth]{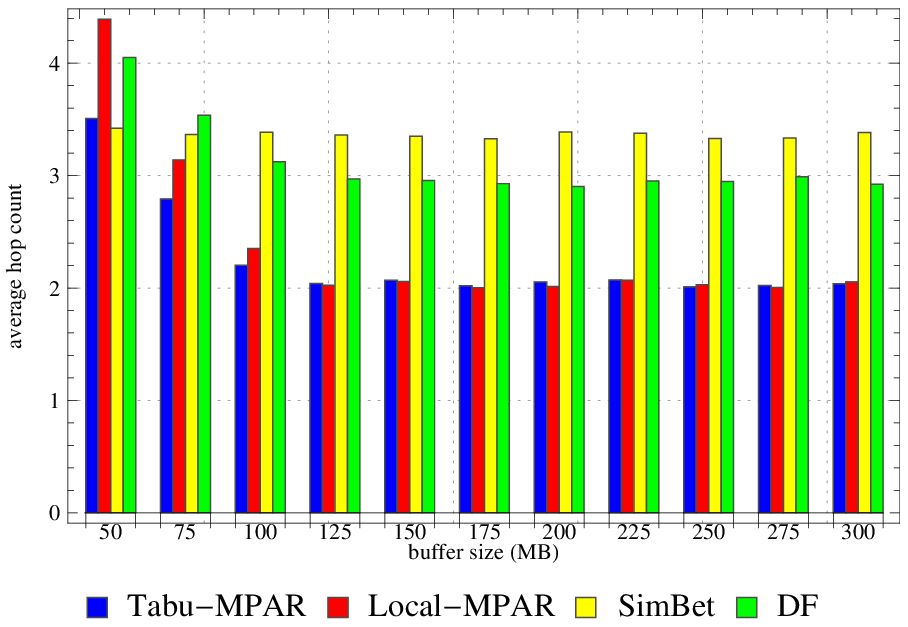}}
\subfigure[Number of nodes: $|\overline{N}|=400$\label{400_avgHop_buffer}]
{\includegraphics[width=0.24\textwidth]{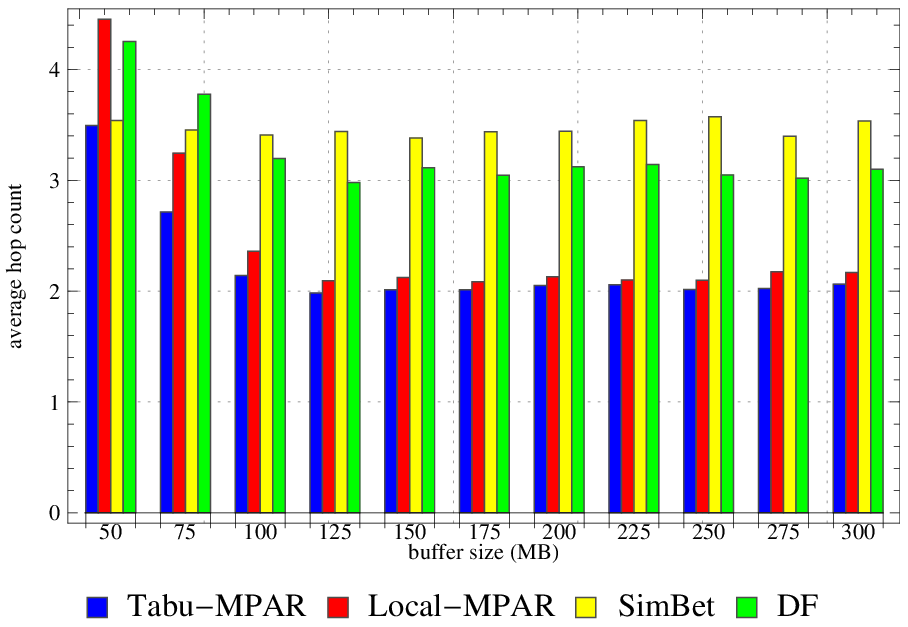}}
\subfigure[Number of nodes: $|\overline{N}|=600$\label{600_avgHop_buffer}]
{\includegraphics[width=0.24\textwidth]{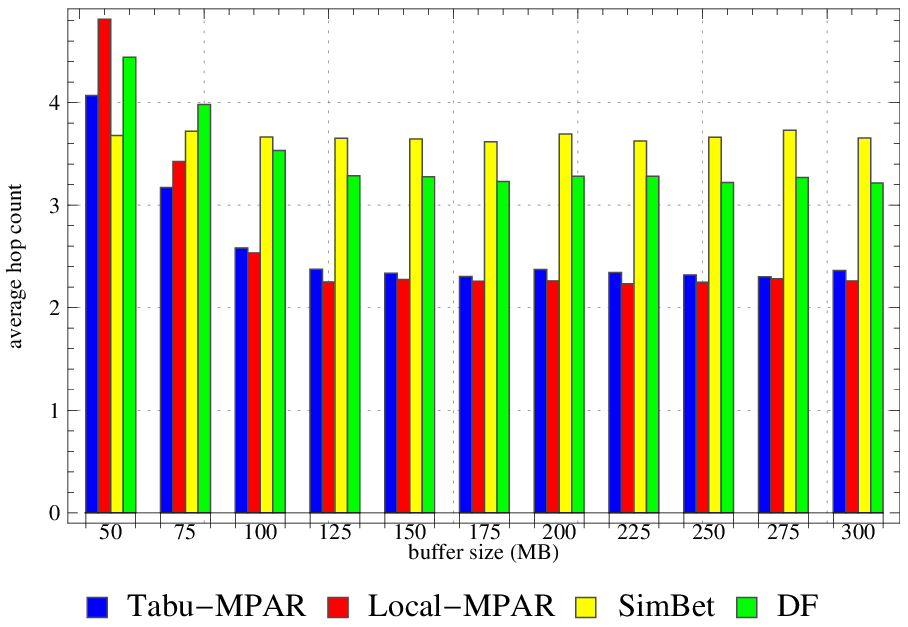}}
\subfigure[Number of nodes: $|\overline{N}|=800$\label{800_avgHop_buffer}]
{\includegraphics[width=0.24\textwidth]{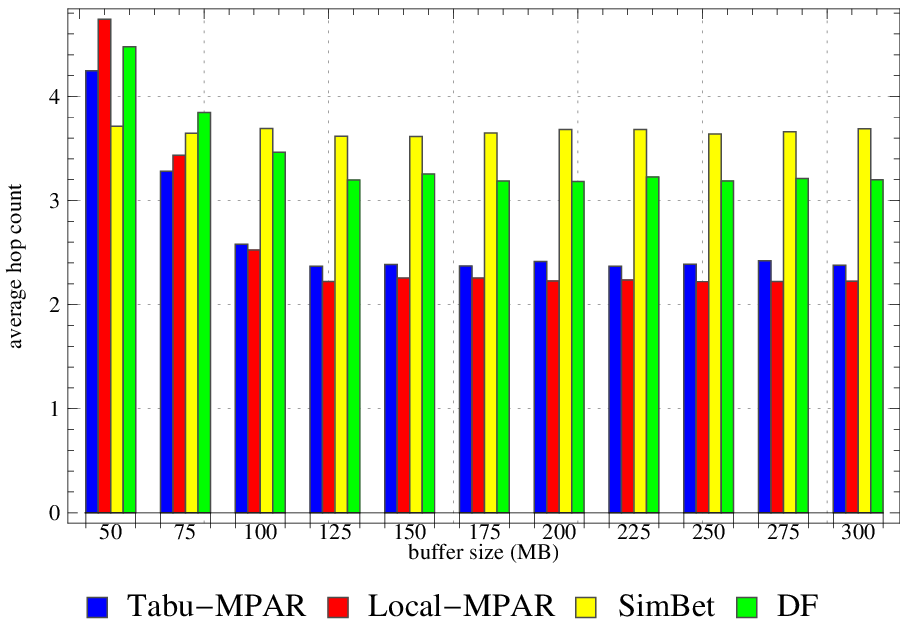}}
\caption{Performance comparisons of delivery probability vs. buffer size}
\label{fig:hop_buffer}
\end{figure*}

We can see from the results that the two MPAR algorithms significantly outperform DF and SimBet.  Compared with DF, Tabu-MPAR and Local-MPAR increase the delivery ratio by about 83.8\% and 40.3\%, and reduces the average latency by about 83.8\% and 57.8\%, respectively. Compared with SimBet, Tabu-MPAR and Local-MPAR increase the delivery ratio by nearly 3 and 2.5 times, and reduces the average latency by about 80\% and 60\%.  Regarding the results in \figurename~\ref{fig:overhead_buffer},Tabu-MPAR and Local-MPAR have almost the same overhead ratio, which is a little bit lower than SimBet and is much lower than DF when the buffer size is smaller than 120M. When there are only a few nodes in a relatively small area, the improvement of the overhead ratio of the social-aware routing schemes is more obvious, as shown in \figurename~\ref{fig:overhead_buffer}. The same as the results in the simulation of varying TTL, we can see from \figurename~\ref{fig:hop_buffer} that the two MPAR algorithms have smaller average hop count than the other two ones. With the number of nodes increasing, the average hop count increases. The reason is as same as in section 6.3.1 that the simulation area is correspondingly magnified and it needs more nodes to cooperate to delivery each message.

\section{Conclusion}
\label{sec:conclusion}
In this paper, we proposed a movement pattern-aware routing protocol MPAR for SDTNs. 
We present a periodical time-aware movement record model and extract the movement pattern from the movement record of nodes, and then each node set is viewed as an entirety during the whole routing process. Two key properties for routing are analyzed and consequently the routing problem is modeled to be a combinatorial optimal search problem. Two search algorithms are proposed to solve the optimization problem, which are respectively based on the local search scheme and the tabu search scheme. Two respective movement pattern-aware routing schemes are designed based on the local search algorithm and the tabu search algorithm, which is called Local-MPAR and Tabu-MPAR and are respectively. In addition, we prove that the Tabu-MPAR can guide the relay node(s) set in evolving to the optimal one. We demonstrate how the MPAR algorithm significantly outperforms the previous ones through extensive simulations, based on the synthetic SDTN mobility model.

\section{Conflict of Interests}
All authors do not have any possible conflict of interests.

\section{Acknowledgments}
This research was supported in part by Natural Science Foundation of Shandong Province under Grant No.ZR2013FQ022 and Foundation research project of Qingdao Science and Technology Plan under Grant No.12-1-4-2-(14)-jch. Science and Technology Plan Project for Colleges and Universities of Shandong Province under Grant NO. J14LN85

\bibliographystyle{elsarticle-num}

\appendix
\newtheorem{app:theorem}{Theorem}
\newtheorem{app:lemma}{Lemma}
\section{Proof of Theorem~1}
\label{app:thm_1}
\begin{app:theorem}
For $\forall k\in[1,q)$, even the $1$---$k$ part of the sequence is known, the newly added element $\mathbb{R}_{k+1}$ can still change the movement pattern to any state with at least 1 non-zero element.
\end{app:theorem}

\begin{proof}
Let us represent each vector $R_i$ as 
\[
R_i=[r^i_1,r^i_2,\ldots,r^i_m]~~~\forall j\in[1,m], r^i_j\geq 0
\]
and the vector $\sum_{i=1}^{i=k}\mathbb{R}_{i}$ as 
\[
\sum_{i=1}^{i=k}\mathbb{R}_{i}=[S^{k}_1,S^{k}_2,\ldots,S^{k}_m]
\]
Thus for $\forall i\in [1,m]$ and $k\in [1,z)$ we have
\[
S^{k+1}_i=S^{k}_i+r^{k+1}_i
\]
Let $\mathcal{P'}=\mathbbm{E}(\sum_{i=1}^{i=k+1}\mathbb{R}_{i})$ and $\mathcal{P}=\mathbbm{E}(\sum_{i=1}^{i=k+1}\mathbb{R}_{i})$ be the old and new movement pattern before and after adding $\mathbb{R}_{k+1}$ respectively. And our objective is to prove that for any possible state of $\mathcal{P}$, there exists a possible $\mathbb{R}_{k+1}$ to transform $\mathcal{P'}$ to $\mathcal{P}$.
From this point, when $\mathcal{P}=[\varkappa _1,\varkappa _2,\ldots,\varkappa _m]$ is known, it is obvious that there exists such an $\mathbb{R}_{k+1}$ if and only if the following linear inequations have solution.
\begin{displaymath}
\forall i\in[1,m]~~~~
\left\{
\begin{array}{ll}
S^{k+1}_i\geq\delta\sum_{j=1}^{j=m}\left(\sfrac{S^{k+1}_j}{m}\right)& \varkappa_i=1 \\
S^{k+1}_i<\delta\sum_{j=1}^{j=m}\left(\sfrac{S^{k+1}_j}{m}\right) & \varkappa_i=0
\end{array}
\right.
\end{displaymath}
i.e.
\[
\left\{
\begin{array}{ll}
\sum_{j=1}^{j=m}r_j^{k+1}-\frac{m}{\delta}r_i^{k+1}\leq \frac{m}{\delta}S_i^{k}-\sum_{j=1}^{j=m}S_j^{k} & \varkappa_i=1 \\
\sum_{j=1}^{j=m}r_j^{k+1}-\frac{m}{\delta}r_i^{k+1}> \frac{m}{\delta}S_i^{k}-\sum_{j=1}^{j=m}S_j^{k} & \varkappa_i=0 \\
\end{array}\right.
\]

For $\forall i$, we denote $\frac{m}{\delta}S_i^{k}-\sum_{j=1}^{j=m}S_j^{k}$ by $w_i$, and represent $w_i$ for $\varkappa_i=1$ and $\varkappa_i=0$ as $w_{i|\varkappa_i=1}$ and $w_{i|\varkappa_i=0}$ respectively. We also represent $r_i^{k+1}$ for $\varkappa_i=1$ and $\varkappa_i=0$ in the similar way. Then we have
\begin{equation}
\sum_{j=1}^{j=m}r_j^{k+1}-\frac{m}{\delta}r_{i|\varkappa_i=0}^{k+1}> w_{i|\varkappa_i=0}  
\label{eq:varkappa_0}
\end{equation}
and
\begin{equation}
\sum_{j=1}^{j=m}r_j^{k+1}-\frac{m}{\delta}r_{i|\varkappa_i=1}^{k+1}\leq w_{i|\varkappa_i=1}
\label{eq:varkappa_1}
\end{equation}
We set $r^{k+1}_{i|\varkappa_i=0}=0$ and $\sum_{j=1}^{j=m}r_j^{k+1}=\max\{w_i\}+\epsilon~(\epsilon>0)$, then the equation~\ref{eq:varkappa_0} always holds for $\forall \varkappa_i=0$.

In the following context we prove that we can choose a suitable $\epsilon$ value to make equation~\ref{eq:varkappa_1} hold. To achieve this, we only need that $\sum_{j=1}^{j=m}r_j^{k+1}-\frac{m}{\delta}r_{i|\varkappa_i=1}^{k+1}\leq \min\{w_i\}$ holds for $\forall \varkappa_i=1$, i.e.
\[
\max\{w_i\}+\epsilon-\frac{m}{\delta}r_{i|\varkappa_i=1}^{k+1}\leq \min\{w_i\}
\]
From above discussion, we only need the following inequations hold simultaneously.
\begin{equation}
r_{i|\varkappa_i=1}^{k+1}\geq\frac{\delta}{m}(\max\{w_i\}-\min\{w_i\}+1)~(0<\delta<1)
\label{eq:q1}
\end{equation}
\begin{equation}
\sum_{j=1}^{j=m}r_{j|\varkappa_j=1}^{k+1}=\max\{w_i\}+\epsilon
\label{eq:q2}
\end{equation}
Assume that the number of $\varkappa_i=1$ in $\mathcal{P}$ is $z$, since that $0<z\leq m$, we have
\[
z\cdot\frac{\delta}{m}(\max\{w_i\}-\min\{w_i\}+\epsilon) \leq  
\delta(\max\{w_i\}-\min\{w_i\}+\epsilon) 
\]
To make equation~\ref{eq:q1} and \ref{eq:q2} hold simultaneously, we should have
\[
\delta(\max\{w_i\}-\min\{w_i\}+\epsilon) < \max\{w_i\}+\epsilon = \sum_{j=1}^{j=m}r_{j|\varkappa_j=1}^{k+1}
\]
i.e.
\[
\epsilon>\frac{\delta}{\delta-1}\min\{w_i\}-\max\{w_i\}
\]
\end{proof}

\section{Proof of Theorem~2}
\label{app:thm_2}
\begin{app:theorem}
Even the global information is known in advance, i.e., the value of $\lambda_{i,j}$ is available for any node $n_i$, the $N_{opt}$ Search Problem is still NP-Hard. 
\label{thm:npc}
\end{app:theorem}
\begin{proof}
We reduce the $N_{opt}$ Search Problem as a Subset Sum Problem (SSP). Assume that we know a solution for a certain SSP instance of which all the elements are denoted by $R={r_1,r_2,\ldots,r_n}$ and the target value is $v$. For clarity, we transform $R$ as follows:
\begin{multline*}
R=\{\log 2^{r_1},\log 2^{r_2},\ldots,\log 2^{r_n}\}\\
=\{\log g_1,\log g_2,\ldots,\log g_n\}
\end{multline*}
where we have $g_i=2^{r_i}~~(i\in[1,n])$.
Let us denote all the non-empty subsets as $R_1,R_2\ldots, R_{2^{|N|-1}}$, and then we construct the corresponding $N_{opt}$ instance, where the target value is $P=v$, and the node set is $N={n_1,n_2,\ldots,n_n}$ corresponding to the above mentioned set $R$. Then each non-empty subset $N_i$ of set $N$ exist in a relationship of one-to-one correspondence with subset $R_i$ of set $R$ (all the subscripts are in one-to-one correspondence), and we denote any subset $N_j$ as 
\[
N_j=\{n_{j_1},n_{j_2},\ldots,n_{j_k}\}
\]
For $\forall j$, we let $\lambda_{i,j}=1$, $\tau_{l}=\sqrt{2}$, and for $\forall n_i\in N-d$, we let $\lambda_{i,j}=\Lambda_i$. Besides, we let all the nodes have the same movement record $\mathbb{R}=[0,0,\ldots,0]$. Consequently we have for all nodes that
\[
\mathcal{P}=[1,1,\ldots,1]
\]

Then we have
\[
P_{N_j,d}=1-\prod_{n_i\in N}(1-\Lambda_{i}e^{-\Lambda_i -1})^n
\]
We denote that $\kappa_i=\left(1-\Lambda_i e^{-\Lambda_i-1}\right)^{n}$, and then we have
\[
P_{N_j,d}=1-\kappa_{j_1}\kappa_{j_2}\cdots\kappa_{j_k}
\]

In this way, there must exist a solution instance of the SSP corresponding to that of the $N_{opt}$ problem, i.e.
\[
v=\log (g_{j_1}g_{j_2}\cdots g_{j_k})\Longleftrightarrow P_{N_j,d}=1-\kappa_{j_1}\kappa_{j_2}\cdots\kappa_{j_n}
\]
Otherwise, if there is no solution for this SSP instance, then there does not exist any solution for the corresponding $N_{opt}$ instance either. This is because that if the $N_{opt}$ instance has a solution, then there exist a set $R_j={r_{j_1},r_{j_2},\ldots,r_{j_n}}$ keeping the following equation hold
\[
P=P_{N_j}
\]
which conflicts the assumption before.
\end{proof}



\end{document}